\documentclass[10pt, leqno]{article}

\usepackage[authoryear,round]{natbib}					
\usepackage{amsmath}    											
\usepackage{graphicx}   											
\usepackage{subfigure}												
\usepackage{amssymb}    											
\usepackage[mathscr]{eucal}										
\usepackage{cancel}														
\usepackage{pstricks}
\usepackage[paperwidth=8.5in,paperheight=11in,top=1.25in, bottom=1.25in, left=1.00in, right=1.00in]{geometry}
\usepackage{mathtools}												
\mathtoolsset{showonlyrefs=true}							
\usepackage{amsthm}														
\newtheorem{theorem}{Theorem}
\newtheorem{lemma}[theorem]{Lemma}						

\theoremstyle{remark}

\newcommand{\<}{\left\langle} 
\renewcommand{\>}{\right\rangle}
\renewcommand{\(}{\left(}				
\renewcommand{\)}{\right)}
\renewcommand{\[}{\left[}
\renewcommand{\]}{\right]}			

\def\E{\mathbb{E}}			
\def\P{\mathbb{P}}										
\def\R{\mathbb{R}}										
\def\I{\mathbb{I}}

\def\A{\mathscr{A}}
\def\B{\mathscr{B}}
\def\D{\mathscr{D}}
\def\F{\mathscr{F}}		
\def\G{\mathscr{G}}		
\def\H{\mathscr{H}}						
\def\L{\mathscr{L}}	
\def\M{\mathscr{M}}			
\def\O{\mathscr{O}}

\def\Ncal{\mathscr{N}}
\def\S{\mathscr{S}}
\def\B{\mathscr{B}}
\def\Xc{\mathscr{X}}

\def\Vc{\mathscr{V}}
\def\Uc{\mathscr{U}}

\def\s{\mathfrak{s}}
\def\m{\mathfrak{m}}

\def\u{\mathfrak{U}}
\def\k{\mathfrak{K}}
\def\eps{\epsilon}
\def\om{\omega}
\def\Om{\Omega}
\def\sig{\sigma}

\def\Lam{\Lambda}
\def\Gam{\Gamma}
\def\lam{\lambda}
\def\del{\delta}
\def\sigb{\overline{\sig}}

\def\thb{\overline{\theta}}
\def\fOmb{\overline{f \! \Om}}

\def\d{\partial}		
\newcommand{\ind}{\perp \! \! \! \perp}

\def\Et{\widetilde{\E}}
\def\Pt{\widetilde{\P}}

\def\Bt{\widetilde{\B}}
\def\Wt{\widetilde{W}}

\numberwithin{equation}{section}	
\numberwithin{theorem}{section}

\begin{document}

\title{Pricing Derivatives on Multiscale Diffusions: an Eigenfunction Expansion Approach}

\author{Matthew J. Lorig\thanks{Work partially supported by NSF grant DMS-0739195.}\\
\emph{ORFE Department, Princeton University}}

\maketitle

\begin{abstract}
Using tools from spectral analysis, singular and regular perturbation theory, we develop a systematic method for analytically computing the approximate price of a large class derivative-assets.  The payoff of the derivative-assets may be path-dependent.  Additionally, the process underlying the derivatives may exhibit killing (i.e., jump to default) as well as combined local/nonlocal stochastic volatility.  The nonlocal component of volatility may be multiscale, in the sense that it may be driven by one fast-varying and one slow-varying factor.  The flexibility of our modeling framework is contrasted by the simplicity of our method.  We reduce the derivative pricing problem to that of solving a single eigenvalue equation.  Once the eigenvalue equation is solved, the approximate price of a derivative can be calculated formulaically.  To illustrate our method, we calculate the approximate price of three derivative-assets: a vanilla option on a defaultable stock, a path-dependent option on a non-defaultable stock, and a bond in a short-rate model.
\end{abstract}


\clearpage

%
%

\section{Introduction}
The spectral representation for the transition density of a general one-dimensional diffusion was obtained in a seminal paper by \citet*{McKean1956}.  Since that time, \emph{spectral theory} -- and more specifically, the study of eigenfunction expansions of linear operators -- has become an essential tool for analysing diffusions.  As a diffusion often serves as the underlying process on which financial models are built, it is not surprising that methods from spectral theory have made their way into mathematical finance as well.
\par
In particular, many problems related to the pricing of derivative-assets have been solved analytically by using methods from spectral theory.  An overview of the spectral method applied to derivative pricing is as follows.  Using risk-neutral pricing, one expresses the value of a derivative-asset $u(t,x)$ as a risk-neutral expectation of some function of the future value of an underlying process $X$.  Mathematically, this is expressed as
\begin{align}
u(t,x) = \Et_x [ H(X_t) ] = \int H(y) \, p(t,x,y) \, dy . \label{eq:u}
\end{align}
Here, $p(t,x,y)$ is the transition density of the $X$ under $\Pt$.  If it turns out that the ininitesmal generator $\L$ of the underlying process is self-adjoint 
\footnote{
An operator $\L$ is \emph{self-adjoint} on a Hilbert space $\H$ with inner product $(\cdot,\cdot)$ if $\text{dom}{(\L)}=\text{dom}(\L^*)$ and $(\L f,g)=(f,\L g)$ for all $f,g \in \text{dom}(\L)$.
Please see appendix \ref{sec:Hilbert} for a brief review of self-adjoint operators in Hilbert Spaces.
}
on a Hilbert space with weighting measure $m(x)dx$ and if the spectrum of $\L$ is purely discrete, then the transition density of $X$ has an eigenfunction expansion
\begin{align}
p(t,x,y) = m(y) \sum_n e^{-\lam_n t} \psi_n(y) \psi_n(x) , \label{eq:p}
\end{align}
where $\left\{ \lam_n \right\}$ are the eigenvalues of $(-\L)$ and $\left\{ \psi_n \right\}$ are the corresponding eigenfunctions
\begin{align}
-\L \, \psi_n = \lam_n \, \psi_n .
\end{align}
The value of a derivative-asset can then be expressed analytically by inserting \eqref{eq:p} into \eqref{eq:u}
\begin{align}
u(t,x)
		&=		\sum c_n \, e^{-\lam_n t} \, \psi_n(x) , &
c_n
		&=		\( \psi_n , H \) := \int H(y) \psi_n(y) m(y) dy .
\end{align}
\par
Under some basic assumptions, the infinitesimal generator of a general one-dimensional diffusion 
\begin{align}
\L			&=				\frac{1}{2} a^2(x) \d^2_{xx} + b(x) \d_x - k(x) , &
x				&\in			(e_1,e_2) , \label{eq:L}
\end{align}
with domain $\text{dom}(\L)$ (described in appendix \ref{sec:BCs})
is \emph{always} self-adjoint on the Hilbert space $\H = L^2(I,\m)$, 
where $I \subset \R$ is an interval with endpoints $e_1$ and $e_2$ and
$\m$ is the speed density of the diffusion
\begin{align}
\m(x)				&:=			\frac{2}{a^2(x)}\exp \( \int_{x_0}^x \frac{2 \, b(y)}{a^2(y)} dy \) .	& &(\text{speed density)} \label{eq:m}
\end{align}
The lower limit of integration $x_0 \in I$ is arbitrary.
Thus, when a one-dimensional diffusion is adequate for describing the dynamics of an underlying, the spectral method outlined above serves as a powerful tool for analytically pricing derivatives on that underlying.  Among the topics that have been addressed by applying spectral methods to one-dimensional diffusions are option pricing (both vanilla and exotic), mortgages valuation, interest rate modeling, volatility modeling, and credit risk (see 
\citet*{linetsky2001, linetsky2002, linetsky2003, linetsky2004, AlbaneseLawi2005, albanese2004unifying, albanese2001black, lewis1998, lipton2002universal, Goldstein1997, linetsky2004blackinterest, linetsky2007intensity, JDCEV, linetsky2004lookback, linetsky2004asian, linetsky2006bankruptcy}).  A useful reference on the topic of spectral methods for one-dimensional diffusions in finance is \citet*{linetskybook}.
\par
As widely applicable as one-dimensional diffusions are in finance, there are applications in which one-dimensional diffusions are not adequate for describing the dynamics of an underlying.  This is the case, for example, in a stochastic volatility setting, where the volatility of the asset that underlies a derivative is controlled by (possibly multiple) nonlocal diffusions.  Ideally, one would like to employ techniques from spectral theory to solve problems that relate to multidimensional diffusions.  Unfortunately, whereas the infinitesimal generator of a one-dimensional diffusion is practically guaranteed to be self-adjoint, the infinitesimal generator of a multidimensional diffusion is only self-adjoint when the drift vector satisfies certain constraints imposed by the volatility matrix.  The drift constraint is not satisfied by any of the most prominent stochastic volatility models -- \citet*{heston}, \citet*{hullwhite1987}, \citet*{stein1991stock} and the SABR model by \citet*{sabr} -- which complicates the use of spectral methods.
\par
Recently, \citet*{lorig2}, show one way to deal with this issue.  By combining techniques from singular perturbation theory and spectral theory, the authors are able to express the approximate price of a (possibly path-dependent) option as an eigenfunction expansion, even though the infinitesimal generator of the two-dimensional diffusion they work with is not self-adjoint.  As notable as their work is, the results of \citet*{lorig2} are valid only when the asset underlying the option is a Black-Scholes-like geometric Brownian motion (GBM) with fast mean-reverting stochastic volatility.
\par
In this paper, we extend the work of \citet*{lorig2} in four important ways.
\begin{enumerate}
\item As a ``base'' model, we work with a general one-dimensional diffusion $dX_t = \nu(X_t) dt + a(X_t)dW_t$.  This is in contrast to \citet*{lorig2}, where the only base model considered is a GBM: $dX_t = \mu X_t dt + \sig X_t dW_t$.
\item The general diffusion we work with may exhibit killing (jump to default) at a rate $h(X_t) \geq 0$.  In the GBM case considered in \citet*{lorig2}, $X$ is always strictly positive.
\item To our general diffusion we add two factors of nonlocal volatility: $a(X_t) \rightarrow a(X_t) f(Y_t,Z_t)$.  The first factor $Y$ is a fast-varying factor.  The second factor $Z$ is slow-varying.  Thus, our model is a \emph{multiscale} stochastic volatility model.  Again, this is in contrast to \citet*{lorig2}, where the analysis is limited to a single fast mean-reverting factor of volatility $\sig X_t \rightarrow f(Y_t) X_t$.
\item In changing from the physical probability measure to the risk-neutral pricing measure, we consider a class of market prices of risk that is general enough to treat credit, equity, and interest rate derivatives in a single framework.  In \citet*{lorig2} the form chosen for the market price of risk restricts the authors to equity derivatives only.
\end{enumerate}
As in \citet*{lorig2}, we will derive an eigenfunction expansion for the approximate price of a derivative-asset despite the fact that the infinitesimal generator we consider is not (in general) self-adjoint.  Unlike \citet*{lorig2}, because our multidimensional diffusion contains both a fast-varying and a slow-varying factor of volatility, we must combine techniques from both singular \emph{and} regular perturbation theory to achieve our result. In \citet*{lorig2}, only singular perturbation techniques are required, due to the presence of a single fast mean-reverting factor of volatility.
\par
Of course, the idea of combining singular and regular perturbation techniques in a multiscale stochastic volatility setting is not particularly new or unique.  The seminal paper on the subject, applied in a Black-Scholes-like GBM setting, is due to \citet*{fouque2004multiscale}.  Further application of the singular and regular perturbation methods developed in \citet*{fouque2004multiscale} led to papers concerning bond-pricing, interest rate derivatives, credit derivatives, and option pricing in a CEV-like setting (see \citet*{foque2008bond,cfps,fouque2006stochastic,fouque2008modeling,fouqueCEV}).  There is also a book by \citet*{fpss}, which contains the many of the key results from the above mentioned publications.  What this paper contributes to the existing literature on multiscale diffusions is flexibility and simplicity.  From a flexibility standpoint, the methods developed in this paper are able to encapsulate, in a unified framework, many of the results contained in
\citet*{fouqueCEV, cfps, foque2008bond, fouque2006stochastic, fouque2004multiscale, fpss, fouque2008modeling},
as well as further results, which are not contained in these works (e.g., jump to default CEV with multiscale stochastic volatility, see section \ref{sec:JDCEV}).  With regards to simplicity, the spectral method we develop reduces the derivative pricing problem to that of solving a single, one-dimensional eigenvalue equation.  Once this equation is solved, the approximate price of a derivative-asset can be calculated formulaically by computing a few simple inner products.  This is in contrast to the methods developed in
\citet*{fouqueCEV, cfps, foque2008bond, fouque2006stochastic, fouque2004multiscale, fpss, fouque2008modeling}, 
where, in order to express the approximate price of a derivative-asset, an inhomogeneous partial differential equation (PDE) must be solved.
\par
The rest of this paper proceeds as follows.  
In section \ref{sec:model} we introduce a class of models described by multiscale diffusions.  We also explain the kind of derivative-asset we wish to consider.
In section \ref{sec:pricing} we solve (approximately), the problem of pricing a derivative-asset.  This is done in several steps.  First, using risk-neutral pricing, we derive a Cauchy problem which, if solved, would yield the exact value of a derivative-asset.  Next, 
we use techniques from singular and regular perturbation theory to formally derive three simpler Cauchy problems, which, if solved, would yield the approximate value of a derivative-asset.  Finally, using eigenfunction expansion techniques, we solve these Cauchy problems explicitly.  The solutions are given in Theorems \ref{thm:u00}, \ref{thm:u10} and \ref{thm:u01}.
In section \ref{sec:examples}, we illustrate our method of pricing derivative-assets with three examples.  We also provide an appendix, which contains some mathematical results that we use throughout this paper.

%
%

\section{A Class of Multiscale Models}\label{sec:model}
Let $(\Om,\F,\P)$ be a probability space supporting correlated Brownian motions $(W^x,W^y,W^z)$ and an exponential random variable $\mathscr{E} \sim \text{Exp}(1)$, which is independent of $(W^x,W^y,W^z)$.  
We shall consider a three-factor economy described by a time-homogenous, continuous-time Markov process $\Xc = (X,Y,Z)$, which takes values in some state space $E = I \times \R \times \R$.  Here, $I$ is an interval in $\R$ with endpoints $e_1$ and $e_2$ such that $-\infty \leq e_1 < e_2 \leq \infty$.  We assume that $\Xc$ starts in $E$ and is instantaneously killed (sent to an isolated cemetery state $\Delta$) as soon as $X$ leaves $I$.  Specifically, the dynamics of $\Xc$ under the physical measure $\P$ are as follows:
\begin{align}
\Xc_t							&=				\begin{cases}
																\( X_t, Y_t, Z_t \)	&\tau_I > t \\
																\Delta									&\tau_I \leq t
																\end{cases}	, &
\tau_I						&=			\inf \left\{ t > 0 : X_t \notin I \right\} ,
\end{align}
where $\(X,Y,Z\)$ are given by
\begin{align}
 \left\{ \begin{aligned}
dX_t				&=			\nu(X_t) \, dt + a(X_t) f(Y_t,Z_t) \, dW^x_t , \\
dY_t				&=			\frac{1}{\eps} \alpha(Y_t) dt + \frac{1}{\sqrt{\eps}} \beta(Y_t) \, dW^y_t , \\
dZ_t				&=			\del c(Z_t) dt + \sqrt{\del} g(Z_t) \, dW^z_t , \\
d\< W^x,W^y\>_t			&=	\rho_{xy} \, dt , \\
d\< W^x,W^z\>_t			&=	\rho_{xz} \, dt , \\
d\< W^y,W^z\>_t			&=	\rho_{yz} \, dt , \\
(X_0,Y_0,Z_0)					&=	(x,y,z) \in E .
\end{aligned} \right. \label{eq:Physical}
\end{align}
Here, $(\rho_{xy},\rho_{xz} ,\rho_{yz})$ satisfy $|\rho_{xy}|, |\rho_{xz}| ,|\rho_{yz}| \leq 1$ and $1 + 2 \rho_{xy} \rho_{xz} \rho_{yz} - \rho_{xy}^2 - \rho_{xz}^2 - \rho_{yz}^2 \geq 0$ so that the correlation matrix of the Brownian motions is positive-semidefinite.
\par
The process $X$ could represent a variety of things.  For example, it could represent the price of a stock, the value of an index, the risk-free short-rate of interest, etc.  More generally, $X$ could represent an exogenous factor that controls the value of any or all of the items mentioned above.  Under the physical measure $\P$, the process $X$ has instantaneous drift $\nu(X_t)$ and stochastic volatility $a(X_t) f(Y_t,Z_t) > 0$, which contains both a local component $a(X_t)$ and nonlocal component $f(Y_t,Z_t)$.  The nonlocal component of volatility $f(Y_t,Z_t)$ is controlled by two factors: $Y$ and $Z$.  We note that the infinitesimal generators of $Y$ and $Z$
\begin{align}
\L_Y^\eps				&=		\frac{1}{\eps} \( \frac{1}{2}\beta^2(y) \, \d^2_{yy}+ \alpha(y) \, \d_y \) , \label{eq:Ly}\\
\L_Z^\del				&=		\del \( \frac{1}{2}g^2(z) \, \d^2_{zz}+ c(z) \, \d_z \) , \label{eq:Lz}
\end{align}
are scaled by factors $1/\eps$ and $\del$ respectively.  Thus, $Y$ and $Z$ have intrinsic time-scales $\eps>0$ and $1/\del>0$.  We assume $\eps<<1$ and $\del << 1$ so that the intrinsic time-scale of $Y$ is small and the intrinsic time-scale of $Z$ is large.  Hence, $Y$ represents a fast-varying factor of volatility and $Z$ represents a slow-varying factor.  
Note that $\L_Y^\eps$ and $\L_Z^\del$ have the form \eqref{eq:L} with $k(x)=0$ for all $x \in I$.
Throughout this paper, we will assume that the domain of any operator of the form \eqref{eq:L} is given by equation \eqref{eq:domainL} of appendix \ref{sec:BCs}.
\par
We are interested in pricing a (possibly defaultable) derivative-asset, whose payoff at time $t>0$ may depend on the path of $X$.
Specifically, we shall consider payoffs of the form
\begin{align}
\text{Payoff}		&=			H(X_t) \, \I_{\left\{ \tau > t \right\}} . \label{eq:Payoff}
\end{align}
Here, $\tau$ is a random time, which represents the default time of the derivative-asset.  Because we are interested in pricing derivatives, we must specify the dynamics of $(X,Y,Z)$ under the risk-neutral pricing measure, which we denote as $\Pt$.  We have the following risk-neutral dynamics
\begin{align}
\left\{ \begin{aligned}
dX_t														&=				\Big( b(X_t) - a(X_t) f(Y_t,Z_t) \Om(Y_t,Z_t) \Big) dt + a(X_t) f(Y_t,Z_t) \, d\Wt^x_t , \\
dY_t														&=				\( \frac{1}{\eps} \alpha(Y_t) - \frac{1}{\sqrt{\eps}} \beta(Y_t) \Lam(Y_t,Z_t) \) dt + \frac{1}{\sqrt{\eps}} \beta(Y_t) \, d\Wt^y_t , \\
dZ_t														&=				\( \del c(Z_t) - \sqrt{\del} g(Z_t) \Gam(Y_t,Z_t) \) dt + \sqrt{\del} g(Z_t) \, d\Wt^z_t , \\
d\< \Wt^x,\Wt^y\>_t			&=				\rho_{xy} \, dt , \\
d\< \Wt^x,\Wt^z\>_t			&=				\rho_{xz} \, dt , \\
d\< \Wt^y,\Wt^z\>_t			&=				\rho_{yz} \, dt , \\
(X_0,Y_0,Z_0)					&=	(x,y,z) \in E ,
\end{aligned} \right. \label{eq:RiskNeutral}
\end{align}
where
\begin{align}
d\Wt_t^x					&:= 			dW_t^x +  \( \frac{ \nu(X_t) - b(X_t) }{a(X_t) f(Y_t,Z_t) } + \Om\( Y_t, Z_t \) \) \, dt, \\
d\Wt_t^y					&:= 			dW_t^y +  \Lam(Y_t,Z_t) \, dt, \\
d\Wt_t^z					&:= 			dW_t^z +  \Gam(Y_t,Z_t) \, dt,
\end{align}
are driftless BM's under $\Pt$.  We assume \eqref{eq:RiskNeutral} has a unique strong solution.
\par
As mentioned above, the random time $\tau$ represents the default time of the derivative-asset.  In our framework, default can occur in one of two ways.  Either default occurs when $X$ exits the interval $I$, or default occurs at a random time $\tau_h$, which is controlled by an instantaneous hazard rate $h(X_t) \geq 0$.  Mathematically, we express the default time $\tau$ as follows
\begin{align}
\left\{
\begin{aligned}
\tau								&=				\tau_I \wedge \tau_h ,  \\
\tau_I							&=				\inf \left\{ t \geq 0 : X_t \notin I \right\} , 	\\
\tau_h							&=				\inf  \left\{ t \geq 0 : \int_0^t h(X_s) \, ds \geq \mathscr{E} \right\}	  ,			&					\mathscr{E} 	&\sim 	\text{Exp}(1) ,		&  &\mathscr{E} 	\ind		(X,Y,Z) .	
\end{aligned}
\right. \label{eq:tau}
\end{align}
Note that the exponentially distributed random variable $\mathscr{E}$ is independent of $(X,Y,Z)$.
\par
Following \citet*{elliot2000}, to keep track of $\tau_h$, we introduce the indicator process $D_t:=\I_{\{t \geq \tau_h \}}$.  Denote by $\mathbb{D}=\{\D_t,t\geq0\}$ the filtration generated by $D$ and by $\mathbb{F}=\{\F_t,t\geq0\}$ the filtration generated by $(W^x,W^y,W^z)$.  Define the enlarged filtration $\mathbb{G}=\{\G_t,t\geq0\}$ where $\G_t=\F_t \vee \D_t$.  Note that $(X,Y,Z)$ is adapted to $\mathbb{G}$ and $\tau$ is a $\mathbb{G}$-stopping time (i.e., $\{\tau \leq t\}\in\G_t$ for every $t \geq 0$).
\par
We shall assume our economy includes a risk-free asset, which grows instantaneously at short-rate $r(X_t) \geq 0$.  Thus, if our economy includes, for example, a non-dividend-paying defaultable asset $S$, whose price process is described by $S_t = \I_{\{\tau>t\}}X_t$, where the state space of $X$ was $I=(0,\infty)$, then the discounted asset price $\{e^{-\int_0^t r(X_s)ds}S_t,t\geq0\}$ must be a $(\Pt,\mathbb{G})$-martingale.  The martingale property can be achieved by setting $b(X_t)=\[r(X_t)+h(X_t)\]X_t$ and $\Om(Y_t,Z_t)=0$ in \eqref{eq:RiskNeutral}.  The reason for adding the hazard rate $h(X_t)$ to the risk-free rate of interest $r(X_t)$ in the drift of $X$ is to compensate for the possibility of a default (see \citet{JDCEV}, Section 2).
\par
On the other hand, if $X$ only describes the risk-free rate of interest through $r(X_t)$, then in changing from the physical measure $\P$ to the pricing measure $\Pt$, one may not have a reason to change the drift of $X$ from $\nu(X_t)$ to $b(X_t)$.  However, one may still 
%
%
wish to consider the effect of including a market price of risk.  In this case, one could set $b(X_t) = \nu(X_t)$ and keep $\Om(Y_t,Z_t) \neq 0$ in \eqref{eq:RiskNeutral}.
\par
We have now described our economy under both the physical and risk-neutral pricing measures, and we have specified the kind of derivative-asset we wish to price.  However, we have not been specific about certain technical assumptions, which we shall need in order to prove the accuracy of our pricing approximation.  
Specific model assumptions can be found in Appendix \ref{sec:assumptions}.

%
%

\section{Derivative Pricing} \label{sec:pricing}
We wish to price a derivative-asset whose payoff is of the form \eqref{eq:Payoff}, where the default time $\tau$ is given by \eqref{eq:tau}.  Using risk-neutral pricing and the Markov property of $\Xc$, the value $u^{\eps,\del}(t,x,y,z)$ of such a derivative-asset at time zero is given by
\begin{align}
u^{\eps,\del}(t,x,y,z)
		&=			\Et_{x,y,z} \[ \exp \( -\int_0^t r(X_s) \, ds \) H(X_t) \, \I_{\left\{ t > \tau \right\}} \] , \label{eq:Expectation}
\end{align}
where $(x,y,z) \in E$ represents the starting point of the process $(X,Y,Z)$.  By conditioning on the path of $X$ (see p. 225 of \citet*{linetskybook}) and by using the Feynman-Kac formula, one can show that $u^{\eps,\del}(t,x,y,z)$ satisfies the following Cauchy problem
\begin{align}
\( -\d_t + \L^{\eps,\del} \) u^{\eps,\del}	&=	0 , 			&	&		(x,y,z) \in E , \, t \in \R^{+} , \label{eq:uPDE} \\
u^{\eps,\del}(0,x,y,z)															&= H(x) , 	
\label{eq:uBC}
\end{align}
where the operator $\L^{\eps,\del}$ is given by 
\begin{align}
\L^{\eps,\del}			&=					\frac{1}{\eps} \L_0 + \frac{1}{\sqrt{\eps}} \L_1 + \L_2 + \sqrt{\frac{\del}{\eps}} \M_3 
																+ \sqrt{\del}  \M_1 + \del \M_2, \label{eq:L,eps,del} \\
\L_0 								&= 					\frac{1}{2} \beta^2(y) \, \d^2_{yy}+ \alpha(y) \, \d_y, \\
\L_1								&= 					\beta(y) \Big( \rho_{xy} a(x) f(y,z) \d_x - \Lam(y,z) \Big) \d_y, \label{eq:L-10} \\
\L_2								&= 					\frac{1}{2} a^2(x) f^2(y,z) \d^2_{xx} + \Big( b(x) - a(x) \Om(y,z) f(y,z) \Big) \d_x - k(x) , \label{eq:L00} \\
\M_3								&= 					\rho_{xz} \beta(y) g(z) \d^2_{yz} , \\
\M_1								&= 					g(z) \Big( \rho_{xz} a(x) f(y,z) \d_x - \Gam(y,z) \Big) \d_z, \\
\M_2								&= 					\frac{1}{2}g^2(z) \d^2_{zz}+ c(z) \d_z , \\
k(x)								&=					r(x) + h(x) .
\end{align}
Aside from the initial condition \eqref{eq:uBC}, the function $u^{\eps,\del}(t,x,y,z)$ must satisfy additional boundary conditions (BCs) at the endpoints $e_1$ and $e_2$ of the interval $I$.  The BCs at $e_1$ and $e_2$ are understood to be contained in the domain of $\L^{\eps,\del}$ and will depend on the nature of the process $\Xc$ near the endpoints of $I$.  Appropriate BCs are discussed in appendix \ref{sec:BCs}.
\par
From equation \eqref{eq:Ly} we see that $\L_0=\L_Y^1$.  We assume that a diffusion with generator $\L_Y^1$ has an invariant distribution $\Pi$ with density $\pi$.  In section \ref{sec:asymptotics}, it will be important to note that the operator $\L_0$ with $\text{dom}(\L_0)=L^2(\R,\pi)$ is self-adjoint acting on the Hilbert space $L^2(\R,\pi)$.


\subsection{Formal Asymptotic Analysis}\label{sec:asymptotics}
We wish to solve Cauchy problem \eqref{eq:uPDE}-\eqref{eq:uBC}.  For general $(f,\alpha,\beta,\Lam,c,g,\Gam)$, no analytic solution exists.  However, we notice that, for fixed $\del$, the terms in \eqref{eq:L,eps,del} containing $\eps$ are diverging in the small-$\eps$ limit, giving rise to a \emph{singular} perturbation.  Meanwhile, for fixed $\eps$, the terms containing $\delta$ are small in the small-$\del$ limit, giving rise to a \emph{regular} perturbation.  Thus, the small-$\eps$ and small-$\del$ regime gives rise to a combined singular-regular perturbation about the $\O(1)$ operator $\L_2$.  This suggests that we seek an asymptotic solution to Cauchy problem \eqref{eq:uPDE}-\eqref{eq:uBC}.  To this end, we expand $u^{\eps,\del}$ in powers of $\sqrt{\eps}$ and $\sqrt{\del}$ as follows
\begin{align}
u^{\eps,\del}
		&=	\sum_{j \geq 0} \sum_{i \geq 0} \sqrt{\eps}^{\,i} \sqrt{\del}^{\,j} u_{i,j} . \label{eq:uexpand}
\end{align}
Our goal will be to find an approximation of the price $u^{\eps,\del} \approx u_{0,0} + \sqrt{\eps} u_{1,0} + \sqrt{\del} u_{0,1}$.  The choice of expanding in half-integer powers of $\eps$ and $\del$ is natural given the form of $\L^{\eps,\del}$.  We will justify this expansion when we prove the accuracy of our pricing approximation in Theorem \ref{thm:accuracy}.
\par
Because we are performing a dual expansion in half-integer powers of $\eps$ and $\del$, we must decide which of these parameters we will expand in first.  We choose to perform a regular perturbation expansion with respect to $\del$ first. Then, within each of the equations that result from the regular perturbation analysis, we will perform a singular perturbation expansion with respect to $\eps$. \footnote{Note that we do not take a limit as $\eps$ and $\del$ go to zero simultaneously.}


\subsubsection*{Regular Perturbation Analysis of Equation \eqref{eq:uPDE}}
The regular perturbation expansion proceeds by separating terms in $\L^{\eps,\del}$ and $u^{\eps,\del}$ by powers of $\sqrt{\del}$
\begin{align}
\L^{\eps,\del}							&=						\L^\eps + \sqrt{\del} \, \M^\eps + \del \M_2 , &
u^{\eps,\del}								&=						\sum_{j \geq 0}  \( \sqrt{\del} \)^j u_j^{\eps} \label{eq:u,del},
\end{align}
where
\begin{align}
\L^\eps 							&=						\frac{1}{\eps} \L_0 + \frac{1}{\sqrt{\eps}} \L_1 + \L_2 , &
\M^\eps								&=						\frac{1}{\sqrt{\eps}} \M_3 +  \M_1 , \label{eq:L,eps}\\
u_j^{\eps}								&=						\sum_{i \geq 0} \( \sqrt{\eps} \)^i u_{i,j} , \label{eq:u,eps}
\end{align}
Inserting expansions \eqref{eq:u,del} into PDE \eqref{eq:uPDE} and collecting terms of like-powers of $\sqrt{\del}$ we find that the lowest order equations of the regular perturbation expansion are
\begin{align}
\O(1):&&
0	
		&=	\(-\d_t + \L^\eps \) u_0^\eps , \label{eq:order1}\\
\O(\sqrt{\del}) :&&
0	
		&=	\(-\d_t + \L^\eps \) u_1^\eps + \M^\eps u_0^\eps . \label{eq:orderhalf}
\end{align}
Now, within equations \eqref{eq:order1} and \eqref{eq:orderhalf}, we will perform a singular perturbation expansion with respect to the parameter $\eps$.  We begin with \eqref{eq:order1}, the $\O(1)$ equation. 
\subsubsection*{Singular Perturbation Analysis of Equation \eqref{eq:order1}}
We insert expansions \eqref{eq:L,eps} and \eqref{eq:u,eps} into \eqref{eq:order1} and collect terms of like-powers of $\sqrt{\eps}$.  The resulting order $\O\( 1/ \eps \)$ and $ \O\( 1/ \sqrt{ \eps} \) $ equations are
\begin{align}
\O\( 1/ \eps \) :& &
0
				&=				\L_0 	 u_{0,0} , \label{eq:eps,-2} \\
\O\( 1/ \sqrt{ \eps} \) :& &
0
				&=				\L_0 	 u_{1,0} +	\L_1 	u_{0,0} . \label{eq:eps,-1}
\end{align}
We note that all terms in $\L_0$ and $\L_1$ take derivatives with respect to $y$.  Therefore, if $ u_{0,0}$ and $ u_{1,0}$ are independent of $y$, equations \eqref{eq:eps,-2} and \eqref{eq:eps,-1} will be satisfied.  Thus, we choose $u_{0,0}=u_{0,0}(t,x,z)$ and $u_{1,0}=u_{1,0}(t,x,z)$.  Continuing the asymptotic analysis, the order $\O\(1 \)$ and $ \O\( \sqrt{ \eps} \) $ equations are
\begin{align}
\O( 1 ) :& &
0
				&=				\L_0 u_{2,0} + \( -\d_t + \L_2 \) u_{0,0} , \label{eq:PoissonA} \\
\O( \sqrt{ \eps} ) :& &
0
				&=				\L_0 u_{3,0} + \L_1 u_{2,0} + \( -\d_t + \L_2 \) u_{1,0} . \label{eq:PoissonF}
\end{align}
where we have used $ \L_1u_{1,0}=0$ in \eqref{eq:PoissonA}.  Equations \eqref{eq:PoissonA} and \eqref{eq:PoissonF} are Poisson equations of the form
\begin{align}
0
				&=				\L_0 u + \chi	. \label{eq:PoissonGeneric}
\end{align}
Recall that $\L_0$ is a self-adjoint operator acting on $L^2(\R,\pi)$.  
By the Fredholm alternative
\footnote{Please refer to Appendix \ref{sec:Poisson} for an discussion of the Fredholm alternative},
in order for equations of the form \eqref{eq:PoissonGeneric} to admit solutions $u \in \text{dom}(\L_0)=L^2(\R,\pi)$, the following \emph{centering condition} condition must be satisfied
\begin{align}
\< \chi \> &:=  \int \chi(y) \, \pi(y) \, dy = 0	, \label{eq:center}
\end{align}
where we have introduced the notation $\< \cdot \>$ to indicate averaging over the invariant distribution $\Pi$.  In equations \eqref{eq:PoissonA} and \eqref{eq:PoissonF} centering condition \eqref{eq:center} corresponds to
\begin{align}
\O(1) :& &
0
				&=				\( -\d_t +\< \L_2 \> \) u_{0,0}, \label{eq:CenterA} \\
\O( \sqrt{ \eps} ) :& &
0
				&=				\< \L_1 u_{2,0}\> + \( -\d_t +\< \L_2 \> \) u_{1,0}. \label{eq:CenterB}
\end{align}
The operator $\< \L_2 \>$ is given by
\begin{align}
\< \L_2 \>					&=				\frac{1}{2} \sigb^2\, a^2(x) \d^2_{xx} + \Big( b(x) - \fOmb\,\, a(x) \Big) \d_x - k(x), &
x										&\in			(e_1,e_2) , \label{eq:<L00>}
\end{align}
where we have defined
\begin{align}
\sigb^2(z) 				&:= 		\< f^2(\cdot,z) \> 		, &
\fOmb(z) 					&:= 		\< f(\cdot,z) \Om(\cdot,z) \> 	, \label{eq:fOmbar}
\end{align}
We assume $\sigb^2(z) < \infty $ and $\fOmb(z) < \infty$.
Given appropriate BCs at $e_1$ and $e_2$, one can find a unique solution $u_{0,0}$ to PDE \eqref{eq:CenterA}.  However, in order to make use of \eqref{eq:CenterB} we need an expression for $\< \L_1 u_{2,0} \>$.  To this end, we note from \eqref{eq:PoissonA} that
\begin{align}
\L_0 u_{2,0}
				&=				- \( -\d_t + \L_2 \) u_{0,0} \\
				&=				- \( -\d_t + \L_2 \) u_{0,0} + \( -\d_t + \< \L_2 \> \) u_{0,0} \\
				&=				- \( \frac{1}{2} a^2 \Big( f^2 - \sigb^2 \Big) \d^2_{xx} - a \Big( f\Om - \fOmb \Big) \d_x \) u_{0,0} . 
									\label{eq:L-20psi20}
\end{align}  
Now, we introduce $\phi(y,z)$ and $\eta(y,z) $ as the solutions to the following Poisson equations
\begin{align}
\L_0 \, \phi				&=				f^2 - \sigb^2 , &
\L_0 \, \eta				&=				f \Om - \fOmb . \label{eq:eta} 
\end{align}
Using \eqref{eq:eta}, we can express $u_{2,0}$ as
\begin{align}
u_{2,0}			&=				- \( \frac{1}{2} a^2 \, \phi \, \d^2_{xx} - a \, \eta \, \d_x \) u_{0,0} + C. \label{eq:psi20}
\end{align}
Note that $C$ is a constant that is independent of $y$.  Now, inserting \eqref{eq:L-10} and \eqref{eq:psi20} into $\< \L_0 u_{2,0}\>$ we find
\begin{align}
\< \L_1 u_{2,0} \>
&=			- \< \( \beta \Big( \rho_{xy}\, a\, f \,\d_x - \Lam \Big) \d_y \) \( \frac{1}{2} \,a^2\, \phi\, \d^2_{xx} - a \eta\, \d_x  \) u_{0,0} \> 
=			- \A \, u_{0,0} . \label{eq:<Lpsi>=A}
\end{align}
The operator $\A$ is given by
\begin{align}
\A
		&=			- \Vc_3  \, a(x) \d_x a^2(x) \d^2_{xx} - \Vc_2 \, a^2(x) \d^2_{xx} - \Uc_2 \, a(x) \d_x a(x) \d_x - \Uc_1 \, a(x) \d_x, \label{eq:A10}
\end{align}
where we have defined four \emph{group parameters}
\footnote{The phrase \emph{group parameter} refers to any $z$-dependent parameter which can be calculated as a moment of model-specific functions.  As we shall see, the effect that the functions ($\beta$, $\Lam$, $g$, $\Gamma$) have on the approximate price of a derivative asset is felt only through eight group parameters.}
\begin{align}
\Vc_3		
		&=			-\frac{\rho_{xy}}{2} \Big\langle \beta f \d_y \phi \Big\rangle, &
\Vc_2
		&=			\frac{1}{2} \Big\langle \beta \Lam \d_y \phi \Big\rangle , &
\Uc_2
		&=			\rho_{xy} \Big\langle \beta f \d_y \eta \Big\rangle , &
\Uc_1
		&=			- \Big\langle \beta \Lam \d_y \eta \Big\rangle .
\end{align}
Inserting \eqref{eq:<Lpsi>=A} into \eqref{eq:CenterB} we find
\begin{align}
\A \, u_{0,0}
			&=		\( -\d_t + \< \L_2 \> \) u_{1,0} . \label{eq:eigentemp}
\end{align}
Given an expression for $u_{0,0}$ and appropriate BCs, one can use PDE \eqref{eq:eigentemp} to find an expression for $u_{1,0}$.  This is as far as we will take the analysis of equation \eqref{eq:order1}.  We now return to the $\O(\sqrt{\del})$ equation \eqref{eq:orderhalf}.  
\subsubsection*{Singular Perturbation Analysis of Equation \eqref{eq:orderhalf}}
The singular perturbation analysis of \eqref{eq:orderhalf} proceeds by inserting expansions \eqref{eq:L,eps} and \eqref{eq:u,eps} into \eqref{eq:orderhalf} and collecting terms of like-powers of $\sqrt{\eps}$.  The resulting order $\O( \sqrt{\del}/ \eps )$ and $ \O( \sqrt{\del}/ \sqrt{ \eps} )$ equations are
\begin{align}
\O( \sqrt{\del}/ \eps ) :& &
0
				&=				\L_0 	u_{0,1} , \label{eq:order-21} \\
\O( \sqrt{\del}/ \sqrt{ \eps} ) :& &
0
				&=				\L_0 u_{1,1} +	\L_1 u_{0,1} , \label{eq:order-11}
\end{align}
where we have used $\M_3 u_{0,0} = 0$.  We note that if $ u_{0,1}$ and $ u_{1,1}$ are independent of $y$, equations \eqref{eq:order-21} and \eqref{eq:order-11} will automatically be satisfied.  Thus, we choose $u_{0,1}=u_{0,1}(x,z)$ and $u_{1,1}=u_{1,1}(x,z)$.  Continuing the asymptotic analysis, the order $\O( \sqrt{\del} )$ equation is
\begin{align}
\O( \sqrt{\del} ) :& &
0
				&=						\L_0 u_{2,1} + \(-\d_t+ \L_2 \) u_{0,1} + \M_1 u_{0,0} ,	\label{eq:PoissonC}
\end{align}
where we have used $\L_1 u_{1,1} = 0$ and $ \M_3 u_{1,0}=0$.  We note that equation \eqref{eq:PoissonC} is a Poisson equation for $u_{2,1}$ of form \eqref{eq:PoissonGeneric}.  By the Fredholm alternative, in order for \eqref{eq:PoissonC} to admit a solution $u_{2,1} \in L^2(\R,\pi)$ centering condition \eqref{eq:center} must be satisfied.  In \eqref{eq:PoissonC} centering condition \eqref{eq:center} corresponds to
\begin{align}
0
		&=		\( -\d_t + \< \L_2 \> \) u_{0,1} + \< \M_1 \> u_{0,0} . \label{eq:EigenTemp2}
\end{align}
Note that $u_{0,0}(t,x,z)$ depends on $z$ only through $\sigb(z)$ and $\fOmb(z)$.  Thus, in \eqref{eq:EigenTemp2} $\< \M_1 \>$ can be written
\begin{align}
\< \M_1 \>
		&=	-	\B\d_z , \label{eq:<L01>} \\
\B
		&=		- \Vc_1 a(x) \d_x - \Vc_0 , &
\Vc_1
		&:=		g \rho_{xz}\<f\> , &
\Vc_0
		&:=		- g \< \Gam \> , \label{eq:B}	\\
\d_z	 
		&= \sigb' \d_{\sigb} + \fOmb' \d_{\fOmb} , &
\sigb' 
		&:= \d_z \sigb , &
\fOmb' 
		&:= \d_z \fOmb . \label{eq:dz}	
\end{align}
Note that we have introduced four more group parameters: $\Vc_1$, $\Vc_0$, $\sigb'$ and $\fOmb'$.
This is as far as we will take the asymptotic analysis of equation \eqref{eq:uPDE}.  For convenience, we review the most important results of this section.


\subsubsection*{Main Results of the Asymptotic Analysis}
\begin{align}
\O ( 1 ):&&
\( -\d_t + \< \L_2 \> \) u_{0,0}
 		&=		0, 		&
u_{0,0}(0,x,z)
		&=		H(x)	,																				\label{eq:EigenA} \\
\O ( \sqrt{ \eps } ):&&
\( -\d_t + \< \L_2 \> \) u_{1,0}
 	&=		\A u_{0,0} , &
u_{1,0}(0,x,z)
		&=		0	\label{eq:EigenB} \\
\O ( \sqrt{ \del } ) :&&
\( -\d_t + \< \L_2 \> \) u_{0,1}
 	&=		\B\d_z u_{0,0} , &
u_{0,1}(0,x,z)
		&=		0 . \label{eq:EigenC} 
\end{align}
The operators $\< \L_2\>$, $\A$, $\B$ and $\d_z$ are defined in \eqref{eq:<L00>}, \eqref{eq:A10}, \eqref{eq:B} and \eqref{eq:dz} respectively.  Note that we have imposed BCs at $t=0$.


\subsection{Explicit Solutions for $u_{0,0}$, $u_{1,0}$ and $u_{0,1}$} \label{sec:u}
In this section we shall explicitly solve equations \eqref{eq:EigenA}, \eqref{eq:EigenB} and \eqref{eq:EigenC} in terms of the eigenfunctions $\{\psi_n \}$ and eigenvalues $\{\lam_n\}$ of the operator $\<\L_2\>$.  To begin, we note that $\<\L_2\>$, given by \eqref{eq:<L00>}, has the form of an infinitesimal generator of a one-dimensional diffusion \eqref{eq:L}
with volatility $\sigb \, a(x)$,
drift $\( b(x) - \fOmb\,\, a(x) \)$
and killing rate $k(x)$.
The $\text{dom}(\<\L_2\>)$ includes BCs, which must be imposed at the endpoints $e_1$ and $e_2$.  Appendix \ref{sec:BCs} describes the appropriate BCs to impose for a general one-dimensional diffusion with a generator of the form \eqref{eq:L}.
\par
Throughout this section we assume $\<\L_2\>$ has a purely discrete spectrum.
We fix a Hilbert space $\H = L^2(I,\m)$ where $\m$ is the speed density corresponding to $\< \L_2 \>$. 
The operator $\<\L_2\>$ is self-adjoint in $\H$ and its domain is a dense subset of $\H$.  Thus, the eigenfunctions $\{ \psi_n \}$ of $\<\L_2\>$ form an orthonormal basis in $\H$.  It is not necessarily true that either $\A : \H \to \H$, $\B  : \H \to \H$ or $\d_z  : \H \to \H$.  As such, we define
\begin{align}
\text{dom}(\A) 				&:= \left\{ \psi \in \H : \A \psi \in \H \right\} , &
\text{dom}(\B) 				&:= \left\{ \psi \in \H : \B \psi \in \H \right\} , \\ 
\text{dom}(\d_z) 			&:= \left\{ \psi \in \H : \d_z \psi \in \H \right\} .
\label{eq:domain,A}
\end{align}


\begin{theorem}\label{thm:u00}
Assume that we can solve the following eigenvalue equation
\begin{align}
- \< \L_2\> \psi_n
		&=		\lam_n \psi_n , &
\psi_n
		&\in \text{\emph{dom}}\(\< \L_2\>\) , \label{eq:Eigen00}
\end{align}
and assume $H \in \H$.  Then the solution $u_{0,0}$ to \eqref{eq:EigenA} is given by
\begin{align}
u_{0,0}
	&=	\sum_n c_n \psi_n T_n , &
c_n
		&=		\( \psi_n , H \) , &
T_n
		&=		e^{ -t \, \lam_n } .
\end{align}
\end{theorem}
\begin{proof}
One can easily verify that $u_{0,0}$ satisfies PDE \eqref{eq:EigenA} assuming \eqref{eq:Eigen00} holds.  To see that the BC $u_{0,0}(0,x,z)=H(x)$ is satisfied, notice that $T_n(0,z)=1$ and apply equation \eqref{eq:Identity} to the payoff function $H$
\begin{align}
\text{Id} \, H
		&=		\sum_n \( \psi_n , H \) \psi_n = \sum_n c_n \psi_n .
\end{align}
\end{proof}


\begin{theorem}\label{thm:u10}
Let $c_n$, $ \psi_n$ and $T_n$ be as described in Theorem \ref{thm:u00} and define
\begin{align}
\A_{k,n}
		&:=		\( \psi_k , \A \psi_n \) , &
U_{k,n}
		&:=		\frac{T_k-T_n}{\lam_k-\lam_n} .
\end{align}
Then the solution $u_{1,0}$ to equation \eqref{eq:EigenB} is
\begin{align}
u_{1,0}
		&= 			\sum_n \sum_{k \neq n} c_n \, \A_{k,n} \psi_k U_{k,n}
						- \sum_n c_n \, \A_{n,n} \psi_n \, t \, T_n . \label{eq:u10,long}
\end{align}
\end{theorem}
\begin{proof}
See appendix \ref{sec:u10proof}.
\end{proof}
\noindent
Note that $u_{1,0}$ is linear in the group parameters $(\Vc_3, \Vc_2, \Uc_2, \Uc_1)$.


\begin{theorem}\label{thm:u01}
Let $c_n$, $ \psi_n$ and $T_n$ be as described in Theorem \ref{thm:u00}, let $U_{k,n}$ be as described in Theorem \ref{thm:u10} and define
\begin{align}
\Bt_{k,n}
		&:=		\( \psi_k , \B\d_z \psi_n \) , &
\B_{k,n}
		&:=		\( \psi_k , \B \psi_n \) , &
V_{k,n}
		&:=		\frac{T_k-T_n}{\( \lam_k - \lam_n \)^2} 
					+ \frac{t \, T_n}{\lam_k - \lam_n} .
\end{align}
Then the solution $u_{0,1}$ to equation \eqref{eq:EigenB} is
\begin{align}
u_{0,1}
		&= 			\sum_n \sum_{k \neq n} c_n \Bt_{k,n} \psi_k U_{k,n} -
						\sum_n c_n \, \Bt_{n,n} \psi_n \, t \, T_n 
						\\ &\qquad +
						\sum_n \sum_{k \neq n} (\d_z c_n )\B_{k,n} \psi_k U_{k,n} -
						\sum_n ( \d_z c_n ) \B_{n,n} \psi_n \, t \, T_n
						\\ &\qquad +
						\sum_n \sum_{k \neq n} c_n \B_{k,n} \psi_k \( \d_z \lam_n \) V_{k,n} + 
						\sum_n c_n \B_{n,n} \psi_n \( \d_z \lam_n \) \tfrac{1}{2} t^2 \, T_n .
						\label{eq:u01,long}
\end{align}
\end{theorem}
\begin{proof}
See appendix \ref{sec:u01proof}.
\end{proof}
\noindent
Note that $u_{0,1}$ is linear in $(\Vc_1 \sigb', \Vc_1 \fOmb', \Vc_0 \sigb', \Vc_0 \fOmb')$.


\subsubsection*{Accuracy of the Pricing Approximation}
We have now  derived an approximation $u^{\eps,\del} \approx u_{0,0} + \sqrt{\eps} \, u_{1,0} + \sqrt{\del} \, u_{0,1}$ for the price of a derivative-asset.  However, this derivation relied on formal singular and regular perturbation arguments.  In what follows, we establish the accuracy of  our approximation.  For our accuracy result, in addition to the assumptions listed in section \ref{sec:assumptions}, we shall need one additional assumption
\begin{itemize}
\item The payoff function $H(x)$ and all of its derivatives are smooth and bounded.
\end{itemize}
Obviously, many common derivatives -- e.g., call and put options -- do not fit this assumption.  To prove the accuracy of our pricing approximation for calls and puts would require regularizing the option payoff as is done in \citet*{fouque2003proof}.  The regularization procedure is beyond the scope of this paper.  As such, we limit our analysis to options with smooth and bounded payoffs.  Our accuracy result is as follows:
\begin{theorem}\label{thm:accuracy}
For fixed $(t,x,y,z)$, there exists a constant $C$ such that for any $\eps \leq 1$, $\del \leq 1$ we have
\begin{align}
\left| u^{\eps,\del} - \( u_{0,0} + \sqrt{\eps} \, u_{0,0} + \sqrt{\del} \, u_{0,1} \) \right| \leq C \( \eps + \del \) .
\label{eq:accuracy}
\end{align}
\end{theorem}
\begin{proof}
See appendix \ref{sec:accuracy}.
\end{proof}
Theorem \ref{thm:accuracy} gives us information about how our pricing approximation behaves as $\eps \to 0$ and $\del \to 0$.  In practice, both $\eps$ and $\del$ are small, but fixed (they do not go to zero).  Without knowing what the constant $C$ is in theorem \ref{thm:accuracy}, it is difficult to gauge exactly how good our pricing approximation is.  As such, in the examples provided in section \ref{sec:examples}, we will compare the approximate prices of derivative-assets (calculated using Theorems \ref{thm:u00}, \ref{thm:u10} and \ref{thm:u01}) to their exact prices (calculated via Monte Carlo simulation).

%
%

\section{Examples} \label{sec:examples}
In this section we compute the approximate price of three derivative-assets: a double-barrier call option, a bond in a short-rate model, and a European call on a defaultable stock.


\subsection{Double-Barrier Call Option with Multiscale Stochastic Volatility}\label{sec:BS}
In our first example, we let $X$ represent the value of a non-dividend paying asset (e.g., a stock, index, etc.).  Often, $X$ is modeled as a GBM with constant volatility (e.g., Black-Scholes).  Here, we model $X$ as a GBM with multiscale stochastic volatility.  Specifically, the $\Pt$ dynamics of $X$ are given by
\begin{align}
dX_t		&=		r X_t \, dt + f\( Y_t, Z_t \) X_t \, d\Wt^x_t  , &
h(X_t)	&=		0 , \label{eq:X,BS}
\end{align}
where $r$ is the risk-free rate of interest and $Y$ and $Z$ are fast- and slow-varying factors of volatility, as described in \eqref{eq:RiskNeutral}.  Note that, as it should be, the discounted price of the asset $\( e^{-rt}X_t \)$ is a martingale under $\Pt$.  We will calculate the approximate price of a double-barrier call option written on $X$.
\par
To start, we use equations \eqref{eq:m} and \eqref{eq:<L00>} to write the operator $\<\L_2\>$ and its associated speed density $\m(x)$
\begin{align}
\< \L_2 \>			&=				\frac{1}{2} \sigb^2 x^2 \d^2_{xx} + r\, x \, \d_x -  r , &
\m(x)									&=				\frac{2}{\sigb^2 x^2 } \exp \( \frac{2 r }{ \sigb^2} \log x \) .\label{eq:LBS}
\end{align}
For a double-barrier call option with knock-out barriers at $L$ and $R$, the option payoff is
\begin{align}
H(X_t) \, \I_{\{ \tau>t \}} &=			\( X_t - K \)^{+} \, \I_{\{ \tau_I > t \}}, &
I														&=			(L,R) , &
0 < L < K < R , \label{eq:DBpayoff}
\end{align}
To calculate the value of this option we must first solve eigenvalue equation \eqref{eq:Eigen00} with $\< \L_2\>$ given by \eqref{eq:LBS} and with BCs
\begin{align}
\lim_{x \searrow L} \psi_n(x)					&=				0 ,	&
\lim_{x \nearrow R} \psi_n(x)					&=				0 .		\label{eq:BS,BC}
\end{align}
Note that we have imposed the regular killing BC at the endpoints $L$ and $R$.  The solution to \eqref{eq:Eigen00} with the above BCs can be found on page 262 of \citet*{linetskybook}
\begin{align}
\psi_n(x)
		&=		\frac{\sigb \sqrt{x}}{\sqrt{\log ( R / L )}} \exp \( \frac{-r}{\sigb^2} \log x\)
					\sin \( \frac{n \pi \log ( x / L )}{\log ( R / L ) } \) , &
n
		&=		1, 2, 3, \cdots , \\
\lam_n
		&=		\frac{1}{2}	\( \frac{n \pi \sigb}{\log (R / L )} \)^2 + \(\frac{\nu^2}{2} + r\) , &
\nu	
		&= 			\frac{r}{\sigb} - \frac{\sigb}{2} . \label{eq:QBS}
\end{align}
Next, we use expressions \eqref{eq:A10} and \eqref{eq:B} to write expressions for the operators $\A$ and $\B$ 
\begin{align}
\A													&=		- \Vc_3  x \, \d_x x^2 \d^2_{xx} - \Vc_2 x^2 \d^2_{xx} , &
\B													&=		- \Vc_1  x \, \d_x - \Vc_0 . \label{eq:A,BS}
\end{align}
Using \eqref{eq:A,BS} it is now straightforward to calculate inner products $\A_{k,n}$, $\B_{k,n}$ and $\Bt_{k,n}$.  For $k \neq n$ we find  
\begin{align}
\A_{k,n}	
		&=		- \Vc_3 \( \frac{\(-1+(-1)^{k+n}\) k n \(4 n^2 \pi^2 \sigb^4+\(-12 r^2+4 r \sigb^2+\sigb^4\) \log^2 (R/L) \)}
														{2 \(k^2-n^2\) \sigb^4 \log^3 (R/L)} \) \\ &\qquad			
					- \Vc_2 \( \frac{4 \(-1+(-1)^{k+n}\) k n r}{\(k^2-n^2\) \sigb^2 \log (R/L)} \) , \\
\B_{k,n}
		&=		\Vc_1 \frac{2 \(-1+(-1)^{k+n}\) k n }{(k-n) (k+n) \log(R/L)}, \\		
\Bt_{k,n}
		&=		- \Vc_1 \sigb' \( \Upsilon_{k,n} \) - \Vc_0 \sigb' 
					\( \frac{8 \(-1+(-1)^{k+n}\) k n r \log(R/L)}{\(k^2-n^2\)^2 \pi^2 \sigb^3} \) , \\
\Upsilon_{k,n}
		&:=		\frac{4 k n r \(\log(L)-(-1)^{k+n} \log(R)\)}{\(k^2-n^2\) \sigb^3 \log (R/L)} \\ &\qquad	
					- \frac{2 \(-1+(-1)^{k+n}\) k n \((k-n) (k+n) \pi^2 \sigb^4-2 r \(-2 r+\sigb^2\) \log^2 (R/L)\)}
									{\(k^2-n^2\)^2 \pi^2 \sigb^5 \log (R/L)} ,
\end{align}
and for $k = n$ we find
\begin{align}
\A_{n,n}	
		&=		- \Vc_3 \( \frac{1}{\sigb^3}\( \frac{3 n^2 \pi ^2 \nu }{\log^2(R/L)}-\nu^3 \) 
					- \frac{1}{\sigb^2} \( \nu^2-\frac{n^2 \pi ^2}{\log^2(R/L)} \) \) 
					- \Vc_2 \( \frac{1}{\sigb^2} \( \nu^2 - \frac{n^2 \pi^2}{\log^2(R/L)} \) + \frac{\nu}{\sigb} \) ,\\
\B_{n,n}
		&=		\Vc_1 \(\frac{2 r - \sigb^2}{2 \sigb^2}\)- \Vc_0 , \\
\Bt_{n,n}
		&=		- \Vc_1 \sigb' \( \frac{1}{2 \sigb} - \frac{r \nu \(\log^2(R)-\log^2(L)\)}{ \sigb^4 \log\(R/L\)} \)
					- \Vc_0 \sigb' \( \frac{1}{\sigb }+\frac{r \(\log^2(R)-\log^2(L)\)}{\sigb^3 \log\(R/L\)} \) .
\end{align}
The calculation of $c_n$ can be found on page 262 of \citet*{linetskybook}
\begin{align}
c_n 		&=		\( \psi_n(\cdot) , ( \cdot - K )^{+} \)	
								=				\frac{L^{\nu/\sigb}}{\log \( R / L\)} \Big( L \, \Phi_n(\nu+\sigb) - K \, \Phi_n(\nu) \Big) , \\
\Phi_n(z)				&:=			\frac{2}{\om_n^2 + z^2} \Big( \exp \( \k z \) \big( \om_n \cos \( \om_n \k\) - z \sin \( \om_n \k \) \big) - \exp \( \u z \)(-1)^n \om_n \Big) , \\
\om_n						&:=					\frac{n \pi}{\u} , \qquad
\k							:=						\frac{1}{\sigb} \log \( \frac{K}{L} \) , \qquad
\u							:=						\frac{1}{\sigb} \log \( \frac{R}{L} \) .
\end{align}
Approximate option prices can now be computed using Theorems \ref{thm:u00}, \ref{thm:u10} and \ref{thm:u01}.
\par
On the left side of figure \ref{fig:DBslow} we plot the approximate price $u_{0,0}+\sqrt{\eps}\,u_{1,0}$ of a double-barrier call option for a specific model that has only a fast-varying factor of volatility.  We suppose the dynamics of $Y$ and the volatility function $f$ are given by
\begin{align}
dY_t
		&=		\( - \frac{1}{\eps} \, Y_t - \frac{1}{\sqrt{\eps}}\beta \, \text{Erf}(Y_t) \) dt + \beta \, d\Wt_t^y , &
f(Y_t)
		&=		\frac{\sigma \, \exp \(Y_t\) }{ \exp\(- \beta^2 / 2\)} , \label{eq:Yexample} \\
\text{Erf}(y)
		&:=		\frac{2}{\sqrt{\pi}} \int_0^y e^{-t^2} dt . 
\end{align}
From comparison we also plot the full price $u^\eps$ (calculated by Monte Carlo simulation) and $u_{0,0}$, which corresponds to the Black-Scholes price with volatility $\sigb$.  On the right side of figure \ref{fig:DBslow} we plot the approximate price $u_{0,0}+\sqrt{\del}\,u_{0,1}$ of a double-barrier call option for a specific model that contains only a slow-varying factor of volatility.  We suppose the dynamics of $Z$ and the volatility function $f$ are given by
\begin{align}
dZ_t
		&=		\( - \del \, Z_t - \sqrt{\del} \, g \, \text{Erf}(Z_t) \) dt + g \, d\Wt_t^z , &
f(Z_t)
		&=		\frac{\sigma \, \exp \( Z_t \)}{\exp \( z \)} . \label{eq:Zexample}
\end{align}
For comparison, we also plot the full price $u^\del$ (calculated by Monte Carlo simulation) and the Black-Scholes price $u_{0,0}$.  As expected, as $\eps$ and $\del$ go to zero, the approximate price converges to the full price, which conveges to the Black-Scholes price.


\subsection{Vasicek Short-Rate with Multiscale Stochastic Volatility}
In our second example, we let $X$ represent the short-rate of interest.  One of the most widely known short-rate models is that of  \citet*{vasicek}, in which $X$ is modeled as an OU process.  Here, we model $X$ as an OU with multiscale stochastic volatility.  Specifically, the $\Pt$ dynamics of $X$ are given by
\begin{align}
dX_t		&=		\Big( \kappa \(\theta -  X_t \) - f(Y_t,Z_t) \Om(Y_t,Z_t) \Big) dt + f\( Y_t, Z_t \) \, d\Wt^x_t  , &
r(X_t)	&=		X_t , &
h(X_t)	&=		0 ,
\end{align}
where $Y$ and $Z$ are fast- and slow-varying factors of volatility, as described in \eqref{eq:RiskNeutral}.  We will calculate the approximate price of zero-coupon bond in this setting.
\footnote{We note that $r(X_t)=X_t$ may become negative when $X$ is described by an OU process.  As such, one may wish to impose a reflecting boundary condition at $x=0$, as carried out in \citet*{linetsky2004blackinterest}.  However, as an OU \emph{without} a reflecting boundary is most prevalent in literature, this is the case we treat here.}
\par
To start, we use equations \eqref{eq:m} and \eqref{eq:<L00>} to write the operator $\<\L_2\>$ and its associated speed density $\m(x)$
\begin{align}
\< \L_2 \>			&=				\frac{1}{2} \sigb^2 \d^2_{xx} + \kappa \( \thb - x \) \d_x -  x , &
\m(x)									&=				\frac{2}{\sigb^2} \exp \( \frac{-\kappa}{\sigb^2} \( \thb -x \)^2 \) , &
\thb										&=				\theta - \tfrac{1}{\kappa}\fOmb .	\label{eq:Lvasicek}
\end{align}
For a zero-coupon bond, the payoff at maturity is
\begin{align}
H(X_t) \, \I_{\left\{ \tau>t \right\}} &=	1 . \label{eq:payoffBond}
\end{align}
\par
In order to price a bond with payoff \eqref{eq:payoffBond}, we must solve eigenvalue equation \eqref{eq:Eigen00} on the interval $I=(-\infty,\infty)$ with $\< \L_2 \>$ given by \eqref{eq:Lvasicek}.  As both $-\infty$ and $\infty$ are natural boundaries, no BCs need to be specified.  The solution to this eigenvalue problem can be found in equation (4.6) of \citet*{linetsky2004blackinterest}
\begin{align}
\psi_n			&=				\Ncal_n \exp \( - A \, \xi - \frac{1}{2} A^2 \) H_n\( \xi + A \), &
\Ncal_n							&=				\( \sqrt{\frac{\kappa}{\pi}} \frac{\sigb}{2^{n+1}n!}\)^{1/2} , \\
A										&=				\frac{\sigb}{\kappa^{3/2}} , &
\xi									&=				\frac{\sqrt{\kappa}}{\sigb} \( x - \thb \) , \\
\lam_n			&=				\lam_n = \thb - \frac{\sigb^2}{2 \kappa^2} + \kappa \, n , &
n										&=				0, 1, 2, \cdots .
\end{align}
Here, $\{H_n\}$ are the (physicists') Hermite polynomials.  
Next, we use \eqref{eq:A10} and \eqref{eq:B} to write expressions for the operators $\A$ and $\B$ 
\begin{align}
\A				
		&=		- \Vc_3  \, \d^3_{xxx} - \( \Vc_2 + \Uc_2 \) \d^2_{xx} - \Uc_1 \, \d_x, &
\B
		&=		- \Vc_1 \, \d_x - \Vc_0 .
\end{align}
It is now straightforward to calculate inner products $\A_{k,n}$, $\B_{k,n}$ and $\Bt_{k,n}$.  Using the recursion relations
\begin{align}
\d_x H_n 						&= 				2 \,n H_{n-1} , &
2 \,x \, H_n				&=				H_{n+1} + \d_x H_n ,
\end{align}
we find
\begin{align}
\A_{k,n}
		&=	- \Vc_3 \left\{
								\sum_{m=0}^{3 \wedge n} \( \begin{array}{c} 3 \\ m \end{array} \) 
								\( \frac{-1}{\kappa} \)^{3-m} \( \frac{2 \sqrt{\kappa}}{\sigb} \)^m
								\frac{n! \, \Ncal_n}{(n-m)! \, \Ncal_{n-m}} \del_{k,n-m} \right\} \\ &\qquad
				- \( \Vc_2 + \Uc_2 \) \left\{
								\sum_{m=0}^{2 \wedge n} \( \begin{array}{c} 2 \\ m \end{array} \) \( \frac{-1}{\kappa} \)^{2-m} 
								\( \frac{2 \sqrt{\kappa}}{\sigb} \)^m
								\frac{n! \, \Ncal_n}{(n-m)! \, \Ncal_{n-m}} \del_{k,n-m} \right\} \\ &\qquad
				- \Uc_1  \left\{
								\( \frac{-1}{\kappa} \) \delta_{k,n} + \( \frac{2\sqrt{\kappa}}{\sigb} \) 
								\frac{ n! \, \Ncal_n}{(n-1)! \Ncal_{n-1}} \delta_{k,n-1} \right\} , \\			
\B_{k,n}
		&= - \Vc_1  \left\{
								\( \frac{-1}{\kappa} \) \delta_{k,n} + \( \frac{2\sqrt{\kappa}}{\sigb} \) 
								\frac{ n! \, \Ncal_n}{(n-1)! \Ncal_{n-1}} \delta_{k,n-1} \right\}
			 -	 \Vc_0 \, \delta_{k,n} , \\
\Bt_{k,n}
		&=	- \Vc_1 \sigb' \left\{
								\[ \( \frac{-1}{\kappa} \) \( \frac{1}{2 \sigb} - \frac{\sigb}{\kappa^3} - \frac{n}{\sigb} \)  \] 
								\del_{k,n} \right. \\ &\left. \qquad	 \qquad						
								+ \[ \( \frac{-1}{\kappa} \) \( \frac{4}{\kappa^{3/2}} \) 
								+ \(\frac{2 \sqrt{\kappa}}{\sigb} \) \( \frac{1}{2 \sigb} - \frac{\sigb}{\kappa^3} - \frac{n}{\sigb} \) \] 
								\frac{n! \Ncal_n}{(n-1)!\Ncal_{n-1}} \del_{k,n-1} \right. \\ &\left. \qquad		\qquad	
								+ \[ \( \frac{-1}{\kappa} \) \( \frac{-2}{\sigb} \) + \(\frac{2 \sqrt{\kappa}}{\sigb} \)  
								\( \frac{4}{\kappa^{3/2}} \) \] \frac{n!\Ncal_n}{(n-2)! \Ncal_{n-2}}\del_{k,n-2} \right. \\ &\left. \qquad \qquad		
								+ \[ \(\frac{2 \sqrt{\kappa}}{\sigb} \) \( \frac{-2}{\sigb} \) \] 
								\frac{n! \Ncal_n}{(n-3)!\Ncal_{n-3}}\del_{k,n-3} \right\} \\ &\qquad																						- 			\Vc_0 \sigb' \left\{ 
								\( \frac{1}{2 \sigb} - \frac{\sigb}{\kappa^3} - \frac{n}{\sigb} \) \del_{k,n}
								+ \( \frac{4}{\kappa^{3/2}} \) \frac{ n! \, \Ncal_{n} }{ (n-1)!\Ncal_{n-1} } \del_{k,n-1} \right.
								\\&\left. \qquad	 \qquad	
								+ \( \frac{-2}{\sigb} \) \frac{ n ! \Ncal_n }{ (n-2)! \Ncal_{n-2} }\del_{k,n-2} \right\} \\ &\qquad
			- \Vc_1 \fOmb' \left\{
								\( \frac{1}{\kappa^3} \) \delta_{k,n} + 
								\( \frac{-4}{\sigb \kappa^{3/2}} \) \frac{n! \Ncal_n}{(n-1)!\Ncal_{n-1}} \delta_{k,n-1} + 
								\( \frac{4}{\sigb^2} \) \frac{n! \Ncal_n}{(n-2)!\Ncal_{n-2}} \delta_{k,n-2} 
								\right\} \\ &\qquad																				
			- \Vc_0 \fOmb' \left\{ 
								\( \frac{-1}{\kappa^2} \) \delta_{k,n} + \( \frac{2}{\sigb \sqrt{\kappa}} \)
								\frac{n! \Ncal_n}{(n-1)!\Ncal_{n-1}} \delta_{k,n-1} \right\} .
\end{align}
The computation of $c_n$ be found on page 63 of in \citet*{linetsky2004blackinterest}
\begin{align}
c_n			&=			\( \psi_n , 1 \) 
				=			\frac{2}{\sigb} \sqrt{ \frac{\pi}{\kappa} } \Ncal_n A^n e^{-A^2 /4} .
\end{align}
The approximate price of a bond can now be calculated using Theorems \ref{thm:u00}, \ref{thm:u10} and \ref{thm:u01}.\\

\subsubsection*{Yield Curve}
For a zero-coupon bond, it is often the yield curve, rather than the bond price itself, that is of fundamental importance.  The yield $R^{\eps,\del}$ of a zero-coupon bond that pays one dollar at time $t$ is defined via the relation
\begin{align}
u^{\eps,\del}			&=			\exp \( -R^{\eps,\del} t \) .
\end{align}
We can obtain an approximation for the yield of a zero-coupon bond by expanding both the bond price $u^{\eps,\del}$ and yield $R^{\eps,\del}$ in powers of $\sqrt{\eps}$ and $\sqrt{\del}$ as follows
\begin{align}
u_{0,0} + \sqrt{\eps} \, u_{1,0}  + \sqrt{\del} \, u_{0,1} + \cdots
												&=			e^{ - \( R_{0,0} + \sqrt{\eps} \, R_{1,0} + \sqrt{\del} \, R_{0,1} + \cdots \) \, t } \\
												&=			e^{ -R_{0,0} t } +  \sqrt{\eps} \( - R_{1,0} \, t \) e^{ -R_{0,0} t } + \sqrt{\del} \( - R_{0,1} \, t \) e^{ -R_{0,0} t } + \cdots .
\end{align}
Matching terms of like-powers of $\sqrt{\eps}$ and $\sqrt{\del}$ we obtain 
\begin{align}
R^{\eps,\del} 			&\approx R_{0,0} + \sqrt{\eps} \, R_{1,0}  + \sqrt{\del} \, R_{0,1}, \\
R_{0,0}					&=			- \tfrac{1}{t}\log \( u_{0,0} \) , \qquad
R_{1,0}						=			\frac{ - u_{1,0} }{ t \, u_{0,0} } , \qquad
R_{0,1}						=			\frac{ - u_{0,1} }{ t \, u_{0,0} } ,
\end{align}
\par
On the left side of figure \ref{fig:Vasicekslow} we plot the approximate yield $R_{0,0}+\sqrt{\eps}\,R_{1,0}$ of a zero coupon bond for a specific model that has only a fast-varying factor of volatility.  We suppose the dynamics of $Y$ and the volatility function $f$ are given by \eqref{eq:Yexample}.  For comparison, we also plot the full yield $R^\eps$ (calculated by Monte Carlo simulation) and the Vasicek yield $R_{0,0}$.  On the right side of figure \ref{fig:Vasicekslow} we plot the approximate yield $R_{0,0}+\sqrt{\del}\,R_{0,1}$ of a zero coupon bond for a specific model that has only a slow-varying factor of volatility.  We suppose the dynamics of $Z$ and the volatility function $f$ are given by \eqref{eq:Zexample}.  For comparison, we also plot the full yield $R^\del$ (calculated by Monte Carlo simulation) and the Vasicek yield $R_{0,0}$.  As expected, as $\eps$ and $\del$ go to zero, the approximate yield converges to the full yield, which converges to the Vasicek yield.


\subsection{Jump to Default CEV with Multiscale Stochastic Volatility}\label{sec:JDCEV}
In our final example, we consider a non-dividend-paying, defaultable asset $S_t = \I_{\{\tau>t\}}X_t$.  As $S$ must be non-negative, we let the state space of $X$ be $(e_1,e_2) = (0, \infty)$.  We base our multiscale diffusion on the jump to default constant elastic variance model (JDCEV) of \citet*{JDCEV}.  Specifically, the $\Pt$ dynamics of $X$ prior to default are given by
\begin{align}
dX_t						&=				\( \mu + c \, X_t^{2\eta} \) X_t \, dt + \( f(Y_t,Z_t) \, X_t^\eta \) X_t \, d\Wt_t^x , &
h(X_t)					&=		 		\mu + c \, X_t^{2 \eta} .
\end{align}
For computational convenience we have set the risk-free interest rate to zero: $r=0$.
The constants $\mu$ and $c$ are assumed to be strictly positive.  As always, $Y$ and $Z$ are fast- and slow-varying factors of volatility, as described in \eqref{eq:RiskNeutral}.  Note that the volatility of $X$ has both a local component $X_t^\eta$ and a nonlocal multiscale component $f(Y_t,Z_t)$.  We assume $\eta<0$ so that the local component of volatility $X_t^\eta$ \emph{increases} as $X_t$ decreases, reflecting the fact that price and volatility are negatively correlated.  The stochastic hazard rate $h(X_t)$ also increases as $X$ decreases, capturing the idea that the probability of default increases as $X$ tends to zero.  Note that $S$ is a $\Pt$-martingale, as it should be.  We will calculate the approximate price of a European put option written on $S$.  The price of a European call option can be obtained through put-call parity.
\par
To begin, we use \eqref{eq:m} and \eqref{eq:<L00>} to write the operator $\<\L_2\>$ and its associated speed density $\m(x)$
\begin{align}
\< \L_2 \>				&=				\frac{1}{2} \sigb^2 x^{2 \eta + 2} \d^2_{xx} + \( \mu + c \, x^{2 \eta} \) x \, \d_x - \( \mu + c \, x^{2 \eta} \),  \label{eq:LJDCEV} \\
\m(x)										&=				\frac{2}{\sigb^2} x^{2c/\sigb^2-2-2\eta} \exp \( A \, x^{-2 \eta} \) , &
A												&=				\frac{\mu}{\sigb^2 | \eta |} 
\end{align}
For the diffusion associated with infinitesimal generator \eqref{eq:LJDCEV} the endpoint $e_2 = \infty$ is a natural boundary.  However, the classification of endpoint $e_1=0$ depends on the values of $\eta$ and $c/\sigb^2$.  The classification is as follows
\begin{align}
c/\sigb^2			&\geq 		1/2					&&\text{and}			&\eta &< 0											,						&&\text{$e_1=0$ is natural,} \\
c/\sigb^2 		&\in			(0,1/2) 		&&\text{and}			&\eta &\in [c/\sigb^2-1/2,0)		,						&&\text{$e_1=0$ is exit,} \\
c/\sigb^2 		&\in			(0,1/2)			&&\text{and}			&\eta &<  c/\sigb^2-1/2				,						&&\text{$e_1=0$ is regular.}\
\end{align}
If the parameters ($c$, $\sigb$, $\eta$) are chosen such that $e_1=0$ is regular, then we specify $e_1=0$ as a killing boundary.
To calculate the approximate price of a European put we must solve the eigenvalue equation \eqref{eq:Eigen00} on the interval $(0,\infty)$ with $\<\L_2 \>$ given by \eqref{eq:LJDCEV} and with the BC
\begin{align}
\lim_{x \searrow 0} \psi_n(x)					&=				0 , &
&\text{if} & 
c/\sigb^2 &\in (0,1/2) .
\end{align}
The solution is given in equation (8.11) of Theorem 8.2 in \citet*{carr}
\begin{align}
\psi_n
		&=	A^{\nu/2} \sqrt{\frac{(n-1)! \,\mu}{\Gam(\nu + n)}}\, x\, \exp\(-A \, x^{-2\eta}\)L_{n-1}^{(\nu)}\( A \, x^{-2 \eta} \)	, &
n		
		&=	1, 2, 3, \cdots ,\label{eq:psi,JDCEV} \\
\lam_n
		&=	2 \mu | \eta | ( n + \nu ) , &
\nu
		&=	\frac{1 + 2 \, (c/\sigb^2)}{2 |\eta|} ,
\end{align}
where $\{L_{n}^{(\nu)}\}$  are the generalized Laguerre polynomials. 
Next, we use \eqref{eq:A10} and \eqref{eq:B} to write expressions for the operators $\A$ and $\B$
\begin{align}
\A		
		&=		- \Vc_3 \, x^{ \eta+1} \d_x x^{2 \eta + 2} \d^2_{xx} - \Vc_2 \, x^{2 \eta + 2} \d^2_{xx}, &
\B
		&=		- \Vc_1 \, x^{ \eta + 1} \d_x	-  \Vc_0 .
\end{align}
Analytic expressions for $\A_{k,n}$, $\B_{k,n}$ and $\Bt_{k,n}$ are easily derived by making the change of variables $A \, x^{-2\eta} \to y$, using $\d_y L_n^{\nu}(y) = - L_{n-1}^{(\nu+1)}(y)$ and
\begin{align}
&\int_0^\infty y^\gamma e^{-y} L_n^{(\alpha)}(y) L_m^{(\beta)}(y) \, dy \\
&\qquad		=		\frac{\Gamma(\alpha-\gamma+n)\Gamma(\beta+1+m)\Gamma(\gamma+1)}{\Gamma(\alpha-\gamma)\Gamma(\beta+1)\,n!\,m!} \,
					{}_3 F_2  \( \begin{array}{ccc}  \gamma+1, & -m, & \gamma+1-\alpha \\ \beta+1, & \gamma+1-\alpha-n , & \end{array} ; 1 \) ,
\end{align}
where where ${}_p F_q$ is a generalized hypergeometric function (the above formula is given in equation (14) of \citet*{laguerre}).
As the formulas for $\A_{k,n}$, $\B_{k,n}$ and $\Bt_{k,n}$ are quite long, for the sake of brevity, we do not provide them here.
\par
The payoff of a European put option with strike price $K>0$ can be decomposed as follows
\begin{align}
(K - S_t)^+		
&=		(K-X_t)^{+} \, \I_{ \left\{ \tau > t \right\}} + K \( 1 - \I_{ \left\{ \tau > t \right\}} \) . \label{eq:PutDecomp}
\end{align}
The first term on the RHS of \eqref{eq:PutDecomp} represents the payoff of a put given no default prior to time $t$.  The second term represents the payoff of a put option given a default occurs prior to time $t$.  Thus, the value of a put option with strike price $K$ -- denoted $u^{\eps,\del}(t,x;K)$ -- can be expressed as the sum of two parts
\begin{align}
u^{\eps,\del}(t,x;K)
																&=		u_0^{\eps,\del}(t,x;K) + u_D^{\eps,\del}(t,x;K) , \label{eq:Put,temp}
\end{align}
where 
\begin{align}
u_0^{\eps,\del}(t,x;K) 		&=	\Et_{x,y,z} \[ (K - X_t)^+ \I_{ \left\{ \tau > t \right\} }\] , \label{eq:u0def} \\
u_D^{\eps,\del}(t,x;K) 		&=	K - K \, \Et_{x,y,z} \[ \I_{ \left\{ \tau > t \right\} } \] \\
													&=	K - K \int_0^\infty  \Et_{x,y,z} \[ \delta_{x'}(X_t) \, \I_{ \left\{ \tau > t \right\} } \] dx' \\
													&=	K - K \, \int_0^\infty u_1^{\eps,\del}(t,x;x') \, dx' , \label{eq:u1integral} \\
u_1^{\eps,\del}(t,x;x')		&=	\Et_{x,y,z} \[ \delta_{x'}(X_t) \, \I_{ \left\{ \tau > t \right\} } \] .	\label{eq:u1def}
\end{align}
Note, because $1 \notin L^2(\R^+,\m)$, we have used the fact that $1 = \int_0^\infty \delta_{x'}(X_t) \, dx'$ on the set $\left\{ \tau > t \right\}$.  This substitution comes at a cost; the integral in \eqref{eq:u1integral} must be computed numerically.  However, numerical evaluation of \eqref{eq:u1integral} is not computationally intensive and does not pose any major difficulties.
\par
Since the payoff functions $H_0(x) = (K-x)^+$ and $H_1(x) = \del_{x'}(x)$ belong to $L^2(\R^+,\m)$, we may calculate
\begin{align}
c_{0,n}									&= 				\( \psi_n(\cdot) , (k-\cdot)^+ \) , &
c_{1,n}										&= 				\( \psi_n , \del_{x'} \) .
\end{align}
The expression for $c_{0,n}$ can be found in equation (8.15) of Theorem 8.4 in \citet*{carr}.  The expression for $c_{1,n}$ is computed trivially.  We have
\begin{align}
c_{0,n}
		&=				\frac{ A^{\nu /2+1} K^{2c/\sigb^2+1-2\eta} \sqrt{ \Gam (\nu+n) } }{ \Gam(\nu +1) \sqrt{ \mu (n-1)! } } \, \times \\
		&					\qquad \[
							\frac{|\eta|}{(c/\sigb^2)+|\eta|} \,{}_2 F_2 \( \begin{array}{cc} 1-n, & \tfrac{c/\sigb^2}{|\eta|}+1 \\ \nu+1, & \tfrac{c/\sigb^2}{|\eta|}+2 \end{array} ; A \, K^{-2\eta }\)
							- \frac{ \Gam(\nu+1)(n-1)! }{ \Gam(\nu + n + 1) } L_{n-1}^{(\nu+1)}(A \, K^{-2\eta})
							\] , \label{eq:c0n} \\
c_{1,n}	
		&= 				\psi_n(x') \, \m(x') .
\end{align} 
The approximate price of a European put option can now be computed using Theorems \ref{thm:u00}, \ref{thm:u10} and \ref{thm:u01}.
\par
For European options, it is often the implied volatility induced by an option price, rather than the option price itself that is of primary interest.  Recall that the implied volatility $I^{\eps,\del}$ of a put option with price $u^{\eps,\del}(t,x;K)$ is defined implicitly through 
\begin{align}
u^{\eps,\del}(t,x;K)=u^{\text{BS}}(t,x,I^{\eps,\del};K)
\end{align}
where $u^{\text{BS}}(t,x,I^{\eps,\del};K)$ is the Black-Scholes price of a put as calculated with volatility $I^{\eps,\del}$.
\par
On the left side of figure \ref{fig:JDCEVslow} we plot the implied volatility induced by the approximate price $u_{0,0}+\sqrt{\eps}\,u_{1,0}$ of a put option for a specific model that has only a fast-varying factor of volatility.  We suppose the dynamics of $Y$ and the volatility function $f$ are given by \eqref{eq:Yexample}.  For comparison, we also plot the implied volatility induced by the full price $u^\eps$ (calculated by Monte Carlo simulation) and the implied volatility induced by the JDCEV price $u_{0,0}$.  On the right side of figure \ref{fig:JDCEVslow} we plot the implied volatility induced by the approximate price $u_{0,0}+\sqrt{\del}\,u_{0,1}$ of a put option for a specific model that has only a slow-varying factor of volatility.  We suppose the dynamics of $Z$ and the volatility function $f$ are given by \eqref{eq:Zexample}.  For comparison, we also plot the implied volatility induced by the full price $u^\del$ (calculated by Monte Carlo simulation) and the implied volatility induced by the JDCEV price $u_{0,0}$.  As expected, as $\eps$ and $\del$ go to zero, the implied volatility induced by the approximate price converges to the implied volatility induced by the full price, which converges to the implied volatility induced by the JDCEV price.

%
%

\section{Review and Conclusions}
This paper develops a general method for obtaining the approximate price for a large class of derivative-assets.  The payoff of the derivatives may be path-dependent and the process underlying the derivative-assets may exhibit jump to default as well as combined local/nonlocal stochastic volatility.  The intensity of the jump to default event may be state-dependent and the nonlocal component of volatility may be multiscale, driven by one fast-varying and one slow-varying factor.
\par
One key advantage of our pricing methodology is that, by combining techniques from spectral theory, singular perturbation theory and regular perturbation theory, we reduce the derivative pricing problem to that of solving a single eigenvalue equation.  Once this equation is solved, the approximate price of a derivative-asset may be calculated formulaically.  We have illustrated the simplicity and flexibility of our method by calculating the approximate prices of thre derivative assets: a double-barrier option on a non-defaultable stock, a European option on a defaultable stock, and a non-defaultable bond in a short-rate model.
\par
We believe that the flexibility of our framework, as well as the analytic tractability that our pricing methodology provides merit further research in this area.  A logical next step, for example, would be to extend the results of this paper to include cases where the eigenvalue equation \eqref{eq:Eigen00} does not have a purely discrete spectrum.

\subsection*{Thanks}
The authors of this paper would like to thank Stephan Sturm, Ronnie Sircar and Jean-Pierre Fouque for helpful conversations.  Additionally, the authors would like to thank two anonymous referees, whose comments vastly improved both the quality and readability of this manuscript.

%
%

\appendix
\section{Appendix}


\subsection{Self-Adjoint Operators acting on a Hilbert Space} \label{sec:Hilbert}
In this appendix we summarize some basic properties of self-adjoint operators acting on a Hilbert space.  A detailed exposition on this topic (including proofs) can be found in \citet*{reedsimon}.  We shall closely follow \citet*{linetskybook}, who provides a more streamlined review.
\par
Let $\H$ be a real, separable
\footnote{A Hilbert space is separable if and only if it admits a countable orthonormal basis (i.e., Schauder basis).}
Hilbert space with inner product $(\cdot,\cdot)$.  A \emph{linear operator}  is a pair $(\text{dom}(\L),\L)$ where $\text{dom}(\L) $ is a linear subset of $\H$ and $\L$ is a linear map $\L:\text{dom}(\L) \to \H$.  
The \emph{adjoint} of an operator $\L$ is an operator $\L^{*}$ such that $(\L f,g)	=	(f, \L^{*} g), \forall \, f \in \text{dom}(\L), g \in \text{dom}(\L^{*})$, where
\begin{align}
\text{dom}(\L^{*}):=\{ g \in \H : \exists \, h \in \H \text{ such that } (\L f, g) = (f,h) \,\, \forall \, f \in \text{dom}(\L) \} .
\end{align}
An operator $(\text{dom}(\L),\L)$ is said to be \emph{self-adjoint} in $\H$ if
\begin{align}
\text{dom}{(\L)}	&=	\text{dom}(\L^*) , &
(\L f,g)					&=	(f,\L g) & \forall \, f,g \in \text{dom}(\L).
\end{align}
Throughout this appendix, for any self-adjoint operator $\L$, we will assume that $\text{dom}(\L)$ is a dense subset of $\H$.
\par
Given a linear operator $\L$, the \emph{resolvent set} $\rho(\L)$ is defined as the set of $\lam \in \mathbb{C}$ such that the mapping $(\L - \text{Id} \, \lam)$ is one-to-one and $R_\lam:=(\L - \text{Id} \, \lam)^{-1}$ is continuous with $\text{dom}( R_\lam ) = \H$.
The operator $R_\lam:\H \to \H$ is called the \emph{resolvent}.  The \emph{spectrum} $\sig(\L)$ of an operator $\L$ is defined as $\sig(\L):= \mathbb{C}\setminus\rho(\L)$.  If $\L$ is self-adjoint, its spectrum is non-empty and real.  We say that $\lam \in \sig(\L)$ is an \emph{eigenvalue} of $\L$ if there exists $\psi \in \text{dom}(\L)$ such that the \emph{eigenvalue equation} is satisfied
\begin{align}
\L \, \psi &= \lam \, \psi . \label{eq:EigenvalueEquation}
\end{align}
A function $\psi$ that solves \eqref{eq:EigenvalueEquation} is called an \emph{eigenfunction} of $\L$ corresponding to $\lam$.  The \emph{multiplicity} of an eigenvalue $\lam$ is the number of linearly independent eigenfunctions for which equation \eqref{eq:EigenvalueEquation} is satisfied.  The spectrum of an operator $\L$ can be decomposed into two disjoint sets called the \emph{discrete} and \emph{essential} spectrum $\sig(\L)=\sig_d(\L) \cup \sig_e(\L)$.  For a self-adjoint operator $\L$, a number $\lam \in \R$ belongs to $\sig_d(\L)$ if and only if $\lam$ is an isolated point of $\sig(\L)$ and $\lam$ is an eigenvalue of finite multiplicity.
\par
The \emph{spectral representation Theorem} is an important tool for analysing self-adjoint operators acting on a Hilbert space.  We state this theorem below in a form which is convenient for the computations in this paper.
\begin{theorem}\label{thm:spectral}
Assume $\L$ is a self-adjoint operator in $\H$ and assume $\L$ has a purely discrete spectrum (i.e., $\sig_e(\L) = \{\emptyset\}$).  The \emph{Spectral Representation Theorem} states that $\L f$ has an eigenfunction expansion
\begin{align}
\L f
		&=		\sum_n \lam_n \, ( \psi_n ,f ) \, \psi_n , &
\forall \, f
		&\in 	\text{\emph{dom}}(\L) ,
\end{align}
where the sum runs over all solutions $\{\lam_n ,\psi_n \}$ of the eigenvalue equation \eqref{eq:EigenvalueEquation}.  Furthermore, for any real-valued Borel-measurable function on $\R$ one can define an operator $\phi(\L)$ using \emph{functional calculus}
\begin{align}
\phi(\L) f
		&:=		\sum_n \phi(\lam_n) \, ( \psi_n ,f ) \, \psi_n , &
\forall \, f
		&\in 	\text{\emph{dom}}(\phi(\L)) , \label{eq:functional} \\
\text{\emph{dom}}(\phi(\L))
		&:= \{ f \in \text{\emph{dom}}(\L) : \sum_n \phi^2(\lam_n) \( \psi_n, f \)^2 < \infty \} .
\end{align}
The operator $\phi(\L)$ is self-adjoint in $\H$ and $\text{\emph{dom}}(\phi(\L)) \subseteq \text{\emph{dom}}(\L)$. 
\end{theorem}
\begin{proof}
See \cite{reedsimon} Theorem VIII.6.
\end{proof}
\noindent
Note that setting $\phi(\lam)=\text{Id}$ yields 
\begin{align}
\text{Id} \, f
		&=		\sum_n ( \psi_n ,f ) \, \psi_n , &
\forall \, f
		&\in  \H , \label{eq:Identity}
\end{align}
which is equivalent to saying that the eigenfunctions $\{ \psi_n \}$ of a densely defined self-adjoint operator in $\H$ form a Schauder basis.  In fact, the basis can be chosen to be orthonormal $(\psi_n,\psi_m)=\delta_{n,m}$.  Also note, setting $\phi(\lam) = R_\lam$ yields an eigenfunction representation of the resolvent operator
\begin{align}
R_\lam \, f
		&=			\sum_n \frac{( \psi_n ,f )}{\lam_n -\lam} \, \psi_n , &
\forall \, f
		&\in  \H , \, \lam \in \rho(\L) .
\end{align}


\subsection{Boundary Conditions} \label{sec:BCs}
According to \citet*{feller1954}, the endpoints $e_1$ and $e_2$ of a one-dimensional diffusion in an interval $I$ can be classified as either \emph{natural}, \emph{exit}, \emph{entrance} or \emph{regular}.  The classification, which can be found in \citet*{borodin,linetskybook}, is done as follows.  For a general infinitesimal generator $\L$ of the form \eqref{eq:L} one can associate a \emph{scale density}
\begin{align}
\s(x)				&:=			\exp \( - \int_{x_0}^x \frac{2 b(y)}{a^2(y)} dy \) ,	& &(\text{scale density)} \label{eq:s}
\end{align}
where the lower limit of integration $x_0 \in (e_1,e_2)$ may be chosen arbitrarily.  From $\s$ one can define a \emph{scale function} $\S$
\begin{align}
\S \( \[ x, y \] \)			&:=			\int_x^y \s(z) \, dz , & 
x,y 										&\in		\( e_1, e_2 \) , \\ 
\S \( \(e_1, y \] \)		&:=			\lim_{x \searrow e_1}	\S \( \[ x, y \] \)	,				&
\S \( \[x, e_2 \) \)		&:=			\lim_{y \nearrow e_2} \S \( \[ x, y \] \) .
\end{align}
Note that the above limits may be infinite.  For some arbitrary $y \in \( e_1, e_2 \)$ we define
\begin{align}
I_1				&:=					\int_{e_1}^y \S \( \( e_1, x \] \) \( 1 + k(x) \) \m(x) \, dx , &
I_2				&:=					\int_y^{e_2} \S \( \[ x, e_2 \) \) \( 1 + k(x) \) \m(x) \, dx , \\
J_1				&:=					\int_{e_1}^y \S \( \[ x, y \] \) \( 1 + k(x) \) \m(x) \, dx , &
J_2				&:=					\int_y^{e_2} \S \( \[ y, x \] \) \( 1 + k(x) \) \m(x) \, dx .
\end{align}
An endpoint $e_i$ is classified as
\begin{itemize}
\item \textbf{Natural} if $I_i = \infty$ and $J_i = \infty$.
No BC needs to be specified at a natural boundary.
The interval $I$ is taken to be open at a natural boundary.
\item \textbf{Exit} if $I_i < \infty$ and $J_i = \infty$.
The appropriate BC at an exit boundary is
\begin{align}
\lim_{x \rightarrow e_i} \psi(x) 														&= 0 .
\end{align}
The interval $I$ is taken to be open at an exit boundary.  
\item \textbf{Entrance} if $I_i = \infty$ and $J_i < \infty$.
The appropriate BC at an entrance boundary is
\begin{align}
 \lim_{x \rightarrow e_i} \frac{\d_x \psi(x)}{\s(x)} 	&= 0 .
\end{align}
The interval $I$ is taken to be open at an entrance boundary.  
\item \textbf{Regular} if $I_i < \infty$ and $J_i < \infty$.
We must specify the behavior of a diffusion at a regular boundary.
Here, we consider only \emph{killing} and \emph{instantaneously reflecting} behavior, for which the appropriate BCs are
\begin{align}
\lim_{x \rightarrow e_i} \psi(x) 														&= 0 \qquad \text{(killing BC)} , &
\lim_{x \rightarrow e_i} \frac{\d_x \psi(x)}{\s(x)} 	&= 0 \qquad \text{(instantaneously reflecting BC)}
\end{align}
The interval $I$ is taken to be open at a regular boundary specified as a killing boundary and closed at a regular boundary specified as instantaneously reflecting.
\end{itemize}
The domain of $\L$ is then
\begin{align}
\text{Dom}\( \L \) 		&=			\left\{ f \in L^2\( I, \m \) : f, \d_x f \in AC_{\text{loc}}(I), \L f \in L^2\( I, \m \) , \text{BCs at $e_1$ and $e_2$} \right\} , \label{eq:domainL}
\end{align}
where $AC_{\text{loc}}(I)$ is the space of functions that are absolutely continuous over each compact subinterval of $I$  (see \citet*{linetskybook}, p. 242).  The BCs at $e_1$ and $e_2$ correspond to the BCs specified above for natural, exit, entrance and regular boundaries.


\subsection{Specific Model Assumptions}\label{sec:assumptions}
\begin{enumerate}
\item We assume existence and uniqueness of $(X,Y,Z)$ as the strong solution to \eqref{eq:Physical}.
\item We assume existence and uniqueness of $(X,Y,Z)$ as the strong solution to \eqref{eq:RiskNeutral}.
\item There exist positive constants $C_\Lam < \infty $ and $C_\Gam < \infty $ such that $|| \Lam ||_\infty <C_\Lam$ and $|| \Gam ||_\infty < C_\Gam$.  
\label{item:GammaLambdaBound}
\item 
\label{item:Y}
Define the time-rescaled process $Y^{(1)}_t:=Y_{\eps \, t}$.  Under $\P$, the process $Y^{(1)}$ has infinitesimal generator $\L_0$.
Under $\P$ we assume:
\begin{enumerate}
\item The process $Y^{(1)}$ is ergodic and has a unique invariant distribution $\Pi$ with density $\pi$.
\item The operator $\L_0$ has a strictly positive spectral gap -- meaning the smallest non-zero eigenvalue $\lam_{min}$ of $\( - \L_0 \)$ is strictly positive.
\item The process $Y^{(1)}$ is reversible -- meaning $\L_0$ is self-adjoint acting on $L^2\( \R, \pi \)$.
\end{enumerate}
These assumptions guarantee (see \citet*{fpss}, p. 93) exponential convergence of $Y^{(1)}$ to its invariant distribution
\begin{align}
\left| \E \[ g\(Y_t^{(1)} \) \] - \<g\> \right| 		&\leq 			C \, \exp \( -\lam_{min} t \) , &
\forall \, g																				&\in				L^2\( \R, \pi \) . &
\end{align}
The above assumptions also ensure (see \citet*{fpss}, p. 139) that for all $k \in \mathbb{N}$ there exists $C(k)<\infty$ such that
\begin{align}
\sup_t \E \[ \left| Y_t^{(1)} \right|^k \] 		&\leq 		C(k)	.			\label{eq:Ybound}
\end{align}
\item Define the time-rescaled process $Z^{(1)}_{t} := Z_{t/\del}$.  Under $\P$, the process $Z^{(1)}$ has infinitesimal generator $\M_2$.
Under $\P$ we assume the process $Z^{(1)}$ admits moments that are uniformly bounded in $s<t$.  That is, for all $k \in \mathbb{N}$ there exists $C(t,k)<\infty$ such that
\begin{align}
\sup_{s \leq t} \E \[ \left| Z_s^{(1)} \right|^k \] 		&\leq 		C(t,k)	.\label{eq:Zbound}
\end{align}
\item \label{item:sig,fOm}
We assume that the functions $f(y,z)$ and $\Om(y,z)$ satisfy $\sigb^2(z)  	<\infty$, $\fOmb(z) 	<\infty$ and
the solutions $\phi(y,z)$ and $\eta(y,z)$ to Poisson equations \eqref{eq:eta} 
are at most polynomially growing.
\item The functions $a(x)$ $b(x)$, $r(x)$ and $h(x)$ satisfy $a>0$, $a \in C^2( I )$, $b \in C^1( I )$, $r \geq 0$, $r \in C(I)$, $h \geq 0$, and $h \in C(I)$.
\item 
\label{item:L}
The spectrum of the operator $\<\L_2\>$, defined in \eqref{eq:<L00>}, 
is simple and purely discrete. 
\end{enumerate}
We note that two of the processes that are most commonly used to model volatility -- the Cox-Ingersoll-Ross (CIR) and Ornstein-Uhlenbeck (OU) processes -- satisfy the assumptions placed on both $Y^{(1)}$ and $Z^{(1)}$.


\subsection{Poisson Equations and the Fredholm Alternative} \label{sec:Poisson}
In this appendix we review the existence and uniqueness of solutions to Poisson equations.  Central to this discussion will be a statement of the Fredholm alternative.  Our presentation follows page 93 of \citet*{fpss}, as well as page 124 of \citet*{fouque2007wave}.
\par
Let $\L$ be a self-adjoint operator densely defined on a real separable Hilbert space $\H$, and let $\{\psi_n,\lam_n\}$ be the complete set of solutions to eigenvalue equation $\L \psi_n = \lam_n \psi_n$.  Consider the following Poisson problem: find, $\psi \in \H$ such that
\begin{align}
\( \L - \lam \) \psi &= \chi , \label{eq:PoissonB}
\end{align}
where the function $\chi \in \H$ and the constant $\lam$ are given.
\begin{theorem}\label{thm:fredholm}
The \emph{Fredholm Alternative} states that one of the following is true:
\begin{enumerate}
\item Either $\lam$ is not an eigenvalue of $\L$, in which case equation \eqref{eq:PoissonB} has a unique solution
\begin{align}
\psi
		&=		R_\lam \, \chi = \sum_n \frac{( \psi_n ,\chi )}{\lam_n -\lam} \, \psi_n .
\end{align}
\item Or, $\lam$ is an eigenvalue of $\L$.  Suppose this is the case.  Let $\lam = \lam_1 = \lam_2 = \cdots = \lam_m$ (i.e., the eigenvalue $\lam$ has multiplicity $m$).  Then \eqref{eq:PoissonB} has a solution if and only if $\( \psi_n, \chi \) = 0$ for all $n \leq m$.  Assuming $\( \psi_n, \chi \) = 0$ for all $n \leq m$, a solution to \eqref{eq:PoissonB} has the form
\begin{align}
\psi
		&=		\sum_{n > m} \frac{( \psi_n ,\chi )}{\lam_n -\lam_k}\, \psi_n
					+ \sum_{n \leq m } c_n \psi_n, &
c_n
		&\in \R .
\end{align}
\end{enumerate}
\end{theorem}
\begin{proof}
See \citet*{reedsimon}, Theorem VI.14 and the ensuing corollary.
\end{proof}
\noindent
Classically, the Fredholm alternative Theorem holds for compact operators on a Hilbert space.  However, the Theorem also holds true for differential operators $\L$ of the form \eqref{eq:L}, with domain \eqref{eq:domainL} acting on the Hilbert space $\H = L^2(I,\m)$, where $\m$ is the speed density corresponding to $\L$ (see section 9.4.2 of \citet*{haberman}).
\par
In particular, we note that $\lam=0$ is an eigenvalue of $\L_0$, which is a self-adjoint operator in $L^2(\R,\pi)$.  The corresponding (normalized) eigenfunction is the constant $\psi_\lam=1$.  Thus, in order for $\L_0 u = \chi$ to have a solution $u \in L^2(\R,\pi)$ we must have $\(1,\chi\)=\int \chi \pi dy =: \< \chi \> = 0$, which is the centering condition \eqref{eq:center}.


\subsection{Proof of Theorem \ref{thm:u10}} \label{sec:u10proof}
We must show that $u_{1,0}$, given by \eqref{eq:u10,long} satisfies PDE and BC \eqref{eq:EigenB}.  It is obvious that the BC $u_{1,0}(0,x,z)=0$ is satisfied.  To show that $u_{1,0}$ satisfies PDE \eqref{eq:EigenB} we note that
\begin{align}
\A u_{0,0}
		&=		\sum_n c_n \( \A \psi_n \) T_n 
		=			\sum_n \sum_k c_n \A_{k,n} \psi_k T_n, \label{eq:A10u00}
\end{align}
where we have used \eqref{eq:Identity} in the second equality.
Now, using \eqref{eq:Eigen00} and 
\begin{align}
\( -\d_t - \lam_k \) U_{k,n} 
		&=		T_n , &
\( -\d_t - \lam_n \) t \, T_n 
		&=		- T_n , \label{eq:dt}
\end{align}
it is easy to show that
\begin{align}
\( -\d_t + \< \L_2 \> \) u_{1,0} = \eqref{eq:A10u00}.
\end{align}


\subsection{Proof of Theorem \ref{thm:u01}} \label{sec:u01proof}
We must show that $u_{0,1}$, given by \eqref{eq:u01,long} satisfies PDE and BC \eqref{eq:EigenC}.  It is obvious that the BC $u_{0,1}(0,x,z)=0$ is satisfied.  To show that $u_{0,1}$ satisfies PDE \eqref{eq:EigenB} we note that
\begin{align}
\B\d_z u_{0,0}
					&= \sum_n c_n \( \B\d_z \psi_n \) T_n 
					+ \sum_n \( \d_z c_n \) \( \B \psi_n \) T_n 
					+ \sum_n	c_n \( \B \psi_n \) \( \d_z T_n \) \\
					&= \sum_n \sum_k c_n \Bt_{k,n} \psi_k T_n 
					+ \sum_n \sum_k ( \d_z c_n ) \B_{k,n} \psi_k T_n 
					- \sum_n	\sum_k c_n \B_{k,n} \psi_k ( \d_z \lam_n ) t \, T_n ,
					\label{eq:Mdzu00}
\end{align}
where we have used \eqref{eq:Identity} in the second equality.  Now, using \eqref{eq:Eigen00}, \eqref{eq:dt} and
\begin{align}
\(-\d_t - \lam_k\)V_{k,n}
		&=	- t \, T_n , &
\(-\d_t - \lam_n\) \tfrac{1}{2} t^2\, T_n
		&=	- t \, T_n
\end{align}
it is easy to show that
\begin{align}
\( -\d_t + \< \L_2 \> \) u_{0,1}
		&=		\eqref{eq:Mdzu00}.
\end{align}


\subsection{Proof of accuracy} \label{sec:accuracy}
Before establishing our main accuracy result -- Theorem \ref{thm:accuracy} -- we shall need the following lemma. 
\begin{lemma}\label{thm:polynomial}
Suppose $J(y,z)$ is at most polynomially growing.  Then, for every $(y,z)$ and $s<t$, there exists a positive constant $C<\infty$ such that for any $\eps \leq 1$ and $\del \leq 1$, we have the following inequality
\begin{align}
\Et_{y,z} \[ \, |  J(Y_s, Z_s) | \, \] &\leq C .
\end{align}
\end{lemma}
\begin{proof}[Proof of Lemma \ref{thm:polynomial}]
It is enough to prove the result for $J(y,z)=y^k$ and $J(y,z)=z^k$ for any $k \in \mathbb{N}$.  We begin with $J(y,z) = z^k$.  Under the physical measure $\P$ we have
\begin{align}
\E \[ | Z_s |^k \] 				&= \E \[ | Z_{\del s}^{(1)} |^k \] 							\leq 				\sup_{\del \leq 1} \E \[ | Z_{\del s}^{(1)} |^k \]  		\leq 				C(s,k)				\leq				C(t,k) ,
\end{align}
by \eqref{eq:Zbound}.  Now define an exponential martingale $M_t^{(\Gam)}$, which relates the dynamics of $Z$ under the risk-neutral measure $\Pt$ to its dynamics under the physical measure $\P$.  We have
\begin{align}
M_t^{(\Gam)}		&:=		\exp \(  -\int_0^t \Gam(Y_s,Z_s) \, dW_s^z - \frac{1}{2} \int_0^t  \Gam^2(Y_s,Z_s) \, ds \) 
								= 		\left. \frac{d\Pt}{d\P} \right|_{\F_t}.
\end{align}
The $\Pt$-expectation of $\left| Z_s \right|^k$ can be found as follows:
\begin{align}
\Et \[ \left| Z_s \right|^k \] 
				&=						\E \[ | Z_s |^k M_s^{(\Gam)} \]	\\
				&=						\E \[ | Z_s |^k \exp \( \frac{1}{2} \int_0^s  \Gam^2(Y_u,Z_u) \, du \) \(M_s^{(2 \Gam)} \)^{1/2} \] \\
				&\leq				\( \E \[ | Z_s |^{2k} \exp \( \int_0^s  \Gam^2(Y_u,Z_u) \, du \) \] \)^{1/2} \( \E \[ M_s^{(2 \Gam)} \] \)^{1/2}								& &\text{(by Cuachy-Schwarz)} \\
				&=						\( \E \[ | Z_s |^{2k} \exp \( \int_0^s  \Gam^2(Y_u,Z_u) \, du \) \] \)^{1/2} 																															& &\text{($M^{(2\Gam)}$ is a $\P$-martingale)} \\
				&\leq				\( \E \[ | Z_{\del s}^{(1)} |^{2k} \] \exp \( s \,  || \Gam ||_\infty^2) \) \)^{1/2} 			\leq 	C ,
\end{align}
where we have used assumption \ref{item:GammaLambdaBound} of section \ref{sec:assumptions} in the last line.  We now examine the case $J(y,z)=y^k$.  We have
\begin{align}
\E \[ | Y_s |^k \] 				&= \E \[ | Y_{s / \eps}^{(1)} |^k \] 							\leq 				\sup_{\eps \leq 1} \E \[ | Y_{s / \eps}^{(1)} |^k \]  		\leq 				C(k)		,
\end{align}
by \eqref{eq:Ybound}.  Using the same argument as above, one can easily show
\begin{align}
\Et \[ \left| Y_s \right|^k \] 
				&=						\E \[ | Y_s |^k M_s^{(\Lam)} \]
				\leq						\( \E \[ | Y_{s / \eps}^{(1)} |^{2k} \] \exp \( s \,  || \Lam ||_\infty^2) \) \)^{1/2} 			\leq 	C ,
\end{align}
which proves the lemma.
\end{proof}
\noindent
We are now in a position to prove Theorem \ref{thm:accuracy}.  We begin by defining a remainder term $R^{\eps,\del}$ by
\begin{align}
u^{\eps,\del} 
		&= u_{0,0} + \sqrt{\eps} \, u_{0,0} + \sqrt{\del} \, u_{0,1} + \eps \( u_{2,0} + \sqrt{\eps} \, u_{3,0} \) + \sqrt{\del} \( \sqrt{\eps} \, u_{1,1} + \eps \, u_{2,1} \) + R^{\eps,\del} .
\end{align}
The functions $u_{0,0}$, $u_{1,0}$ and $u_{0,1}$ are the unique solutions to \eqref{eq:EigenA}, \eqref{eq:EigenB} and \eqref{eq:EigenC} respectively.  The function $u_{2,0}$ is given by \eqref{eq:psi20}.  And $u_{2,0}$ is the solution to Poisson equation \eqref{eq:PoissonF}.  To characterize $u_{1,1}$ and $u_{2,1}$ we must continue the singular perturbation analysis of equation \eqref{eq:orderhalf} a bit further.  The $\O(\sqrt{\eps \, \del})$ equation that results from continuing the asymptotic analysis is
\begin{align}
\O( \sqrt{\eps \, \del} ) :& &
0
		&=		\L_0 u_{3,1} + \L_1 u_{2,1} + \(-\d_t + \L_2\) u_{1,1} + \M_3 u_{2,0} + \M_1 u_{1,0} \label{eq:PoissonD}
\end{align}
Equation \eqref{eq:PoissonD} is a Poisson equation of the form \eqref{eq:PoissonGeneric}.  In order for \eqref{eq:PoissonD} to admit a solution $u_{3,1}$ in $L^2(\R,\pi)$, centering condition \eqref{eq:center} must in satisfied.  In \eqref{eq:PoissonD} the centering condition corresponds to
\begin{align}
0
		&=		\< \L_1 u_{2,1} \> + \(-\d_t + \< \L_2 \> \) u_{1,1} + \< \M_3 u_{2,0} \> + \< \M_1 \> u_{1,0} . \label{eq:RandomName}
\end{align}
Now, by introducing $\xi(y,z)$ and $\zeta(y,z)$ as solutions to
\begin{align}
\L_0 \xi					&=				f - \< f \> ,  &
\L_0 \zeta				&=				\Gam - \< \Gam \> . \label{eq:zeta} 
\end{align}
and by subtracting \eqref{eq:EigenTemp2} from \eqref{eq:PoissonC}, we can express $u_{2,1}$ as
\begin{align}
u_{2,1}	
		&=	- \( \frac{1}{2} a^2\, \phi\, \d^2_{xx} - a\, \eta \, \d_x \) u_{0,1} 
				- g \, \Big( \rho_{xz}\, a \, \xi \, \d_x - \zeta \, \Big) \d_z u_{0,0} + D , \label{eq:Psi21}
\end{align}
where $D(x,z)$ is a constant which is independent of $y$.  Substituting \eqref{eq:Psi21} into \eqref{eq:RandomName} characterizes $u_{1,1}$ in terms of $u_{0,0}$, $u_{1,0}$, $u_{1,0}$ and $u_{0,1}$.  We choose $u_{1,1}$ as the solution to \eqref{eq:RandomName} with BC $u(0,x,z)=0$.
\par
Now, we compute
\begin{align}
0			&=			\( -\d_t + \L^{\eps,\del} \) u^{\eps,\del} \\
			&=			\( -\d_t + \L^{\eps,\del} \) R^{\eps,\del}
							+ \frac{1}{\eps} F_0 + \frac{1}{\sqrt{\eps}} F_1 + F_2
							+ \sqrt{\del} \( \frac{1}{\eps} F_3 + \frac{1}{\sqrt{\eps}} F_4	 + F_5 \) \\
			&\qquad	+ \eps \, R_1^\eps + \sqrt{\eps \, \del} \, R_2^\eps + \del R_3^\eps , \label{eq:RandF}
\end{align}
where
\begin{align}	
F_0		&=		\L_0 u_{0,0} , \\		
F_1		&=		\L_0 u_{1,0} + \L_1 u_{0,0} , \\	
F_2		&=		\L_0 u_{2,0} + \L_1 u_{1,0} + \( -\d_t + \L_2 \) u_{0,0} , \\		
F_3		&= 		\L_0 u_{0,1} , \\		
F_4		&= 		\L_0 u_{1,1} + \L_1 u_{0,1} + \M_3 u_{0,0},  \\		
F_5		&= 		\L_0 u_{2,1} + \L_1 u_{1,1} + \M_3 u_{1,0} + \M_1 u_{0,0} +  \( -\d_t + \L_2 \) u_{0,1} ,
\end{align}
and
\begin{align}
R_1^\eps 	&=						\( -\d_t + \L_2 \) u_{2,0} + \L_1 u_{3,0} + \sqrt{\eps} \( -\d_t + \L_2 \) u_{3,0}, \\
R_2^\eps 	&=						\( -\d_t + \L_2 \) u_{1,1} + \L_1 u_{2,1} + \M_1 u_{1,0} + \M_3 u_{2,0} \\
								&\qquad 		+ \, \sqrt{\eps} \(  \( -\d_t + \L_2 \) u_{2,1} + \M_1 u_{2,0} + \M_3 u_{3,0} \) + \eps \, \M_1 u_{3,0}, \\
R_3^\eps 	&=						\M_1	u_{0,1} + \M_2	u_{0,0} + \M_3	u_{1,1} + \sqrt{\eps} \( \M_1	u_{1,1} + \M_2	u_{1,0} + \M_3	u_{2,1} \)	\\
								&\qquad		+ \, \eps \( \M_1	u_{2,1} + \M_2	u_{2,0} \).
\end{align}
From the choices made in section \ref{sec:asymptotics}, it is straightforward to show 
$F_0	=	F_1	=	F_2	= F_3	= F_4	= F_5	=	0$.
Hence, from \eqref{eq:RandF} we have
\begin{align}
0		
		&=	\( -\d_t + \L^{\eps,\del} \) R^{\eps,\del} + \eps \, R_1^\eps + \sqrt{\eps \, \del} \, R_2^\eps 
				+ \del R_3^\eps , \label{eq:RPDE} \\
R(0,x,y,z)
		&=	\eps \, G_1^\eps(x,y,z) + \sqrt{\eps \, \del} \, G_2^\eps(x,y,z)	,  \label{eq:RBC}
\end{align}
where
\begin{align}
G_1^\eps(x,y,z)
		&:=		- u_{2,0}(0,x,y,z) - \sqrt{\eps} \, u_{3,0}(0,x,y,z) ,  \\
G_2^\eps(x,y,z)		
		&:=		- u_{1,1}(0,x,y,z) - \sqrt{\eps} \, u_{2,1}(0,x,y,z)	.
\end{align}
Using the Feynman-Kac formula, we can express $R^{\eps,\del}(t,x,y,z)$, which is the solution to PDE \eqref{eq:RPDE} with BC \eqref{eq:RBC}, as an expectation
\begin{align}
R^{\eps,\del}(t,x,y,z)	
			&=						\eps \, \Et_{x,y,z} \[ e^{-\int_0^t k(X_s) ds} G_1^\eps(X_t,Y_t,Z_t) + \int_0^t e^{-\int_0^s k(X_u) du} R_1^\eps(s, X_s,Y_s,Z_s) \, ds \] \\
			&\qquad 		+ \, \sqrt{\eps\,\del} \, \Et_{x,y,z} \[ e^{-\int_0^t k(X_s) ds} G_2^\eps(X_t,Y_t,Z_t) + \int_0^t e^{-\int_0^s k(X_u) du} R_2^\eps(s, X_s,Y_s,Z_s) \, ds \] \\
			&\qquad 		+ \, \del \, \Et_{x,y,z} \[  \int_0^t e^{-\int_0^s k(X_u) du} R_3^\eps(s, X_s,Y_s,Z_s) \, ds \] .
\end{align}
From the assumptions of section \ref{sec:assumptions} one can deduce that the functions $\( R_1^\eps, R_2^\eps, R_3^\eps, G_1^\eps, G_2^\eps \)$ are bounded in $x$ and at most polynomially growing in $(y,z)$ (see \citet*{fpss}).  Hence, by Lemma \ref{thm:polynomial} we have
\begin{align}
\left| R^{\eps,\del} \right|			&\leq			\eps \, C_1 + \sqrt{\eps\,\del} \, C_2 + \del \, C_3 \leq \( \eps + \del \) \, C_4 .
\end{align}
Finally
\begin{align}
&\left| u^{\eps,\del} - \( u_{0,0} + \sqrt{\eps} \, u_{1,0}+ \sqrt{\del} \, u_{0,1}\)\right| \\
			&\qquad \leq				\left| R^{\eps,\del} \right| +  \left| \eps \, u_{2,0} + \eps^{3/2} u_{3,0} + \sqrt{\eps \, \del} \, u_{1,1} + \eps \sqrt{\del} \, u_{2,1} \right| \\
			&\qquad \leq				\( \eps + \del \) C_4 + \eps \, \left| u_{2,0} + \sqrt{\eps} \, u_{3,0}  \right| + \sqrt{\eps \, \del} \, \left| u_{1,1} + \sqrt{\eps} \, u_{2,1}  \right| \\
			&\qquad \leq				\( \eps + \del \) C ,
\end{align}
which is the claimed accuracy result.

%
%

\clearpage
\bibliographystyle{chicago}
\bibliography{Bibtex-Master-Multiscale}

%
%

\clearpage


\begin{figure}
\centering
\begin{tabular}{ | c | c |}
\hline
\includegraphics[width=.5\textwidth,height=.25\textheight]{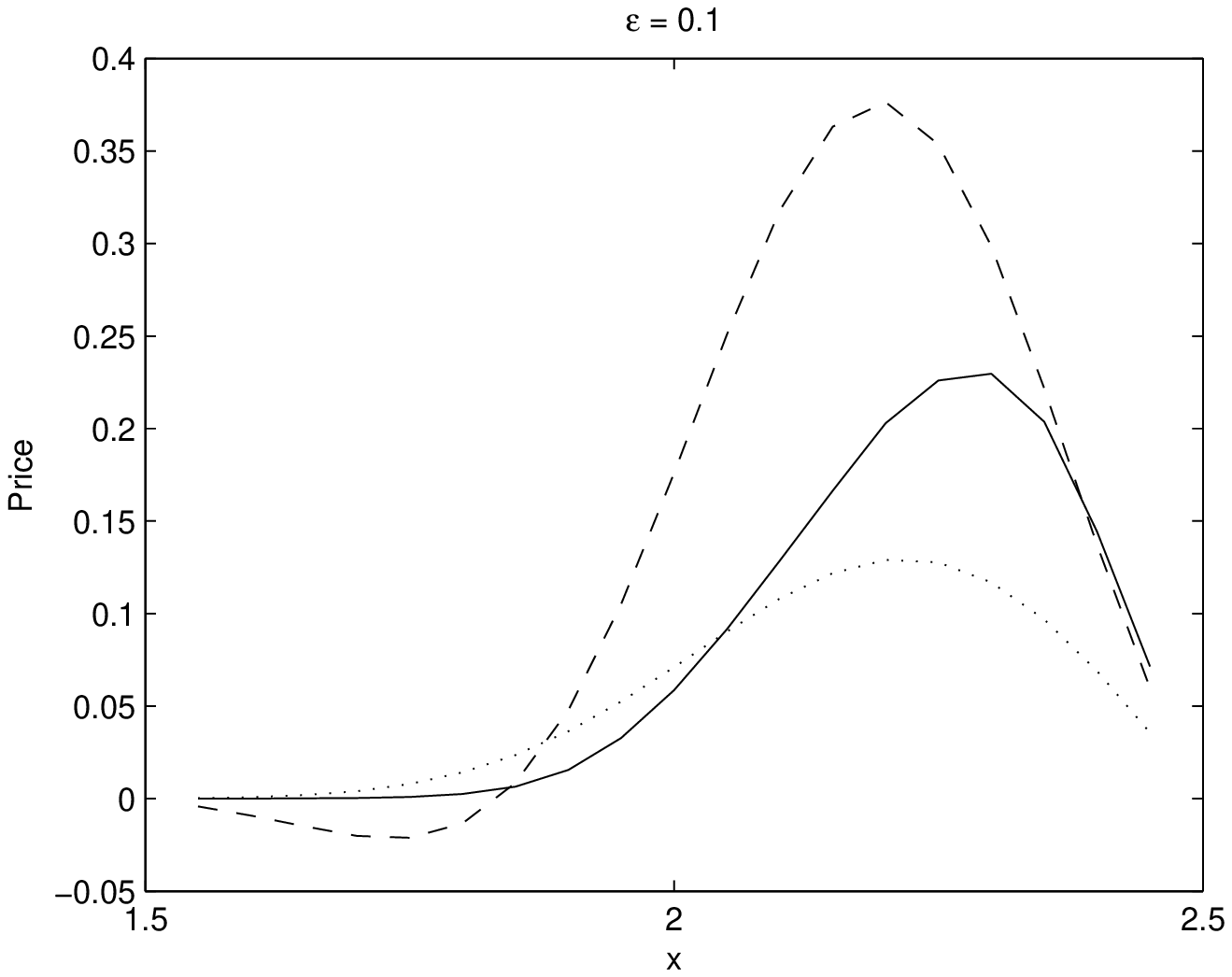} &
\includegraphics[width=.5\textwidth,height=.25\textheight]{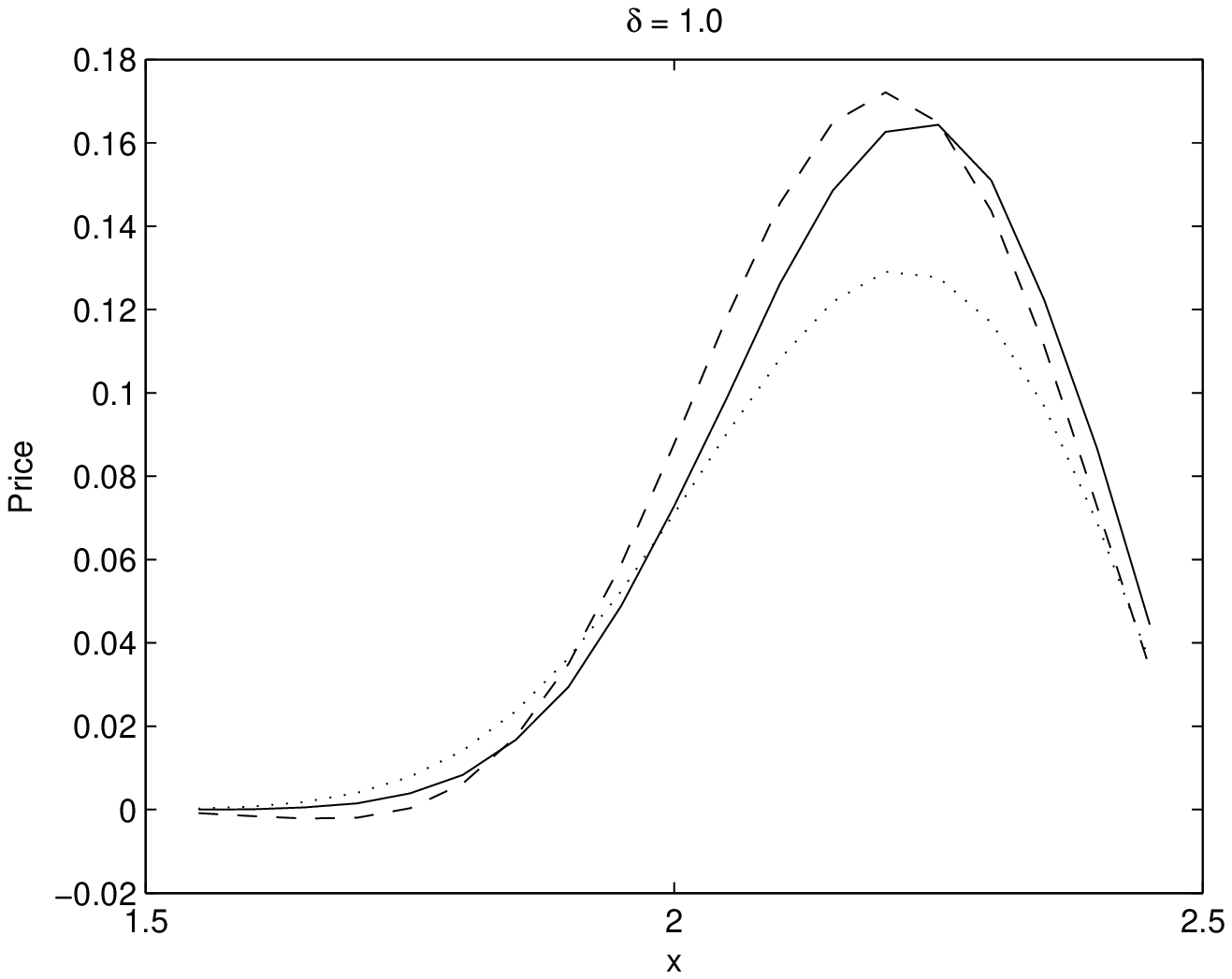} \\ \hline
\includegraphics[width=.5\textwidth,height=.25\textheight]{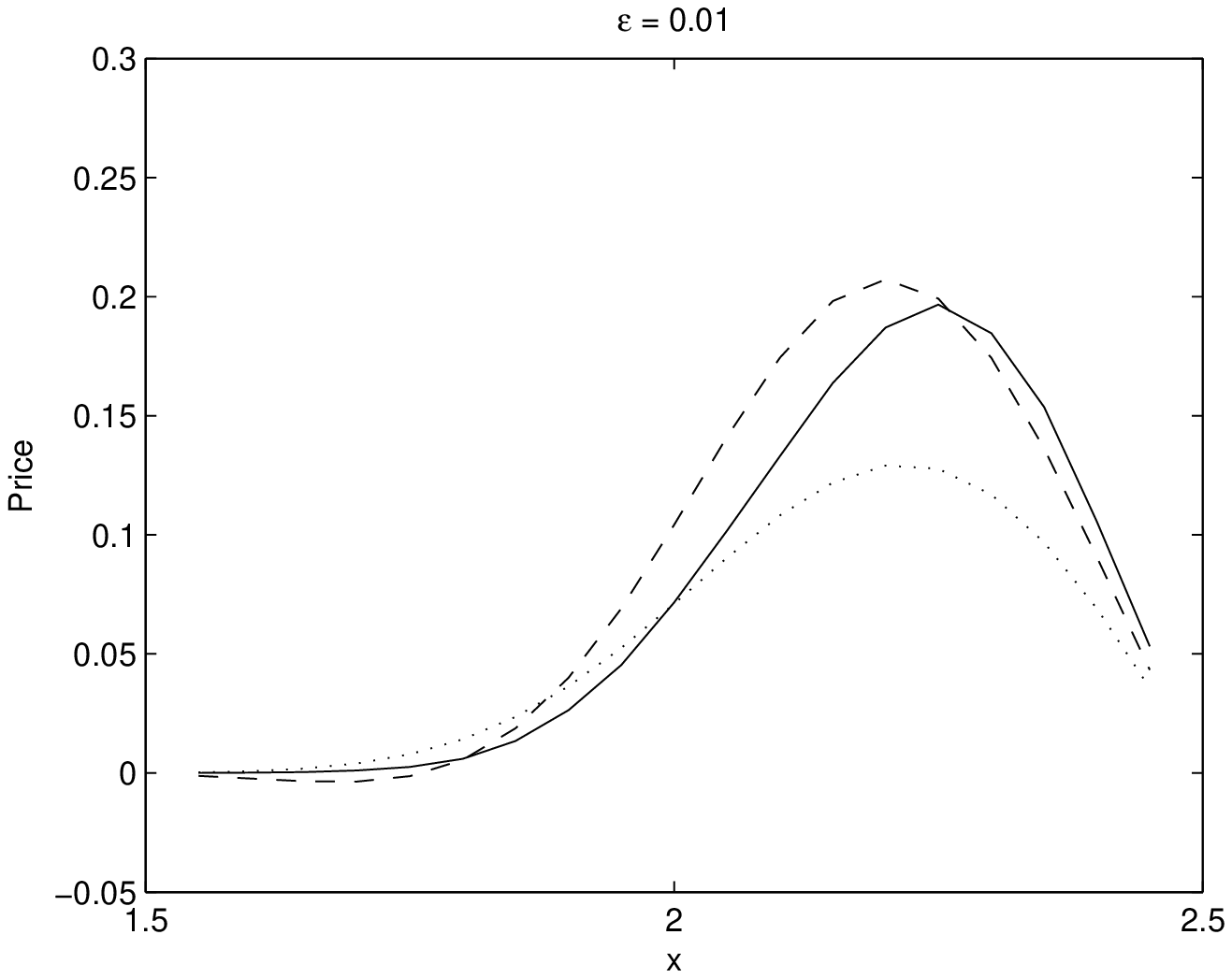} &
\includegraphics[width=.5\textwidth,height=.25\textheight]{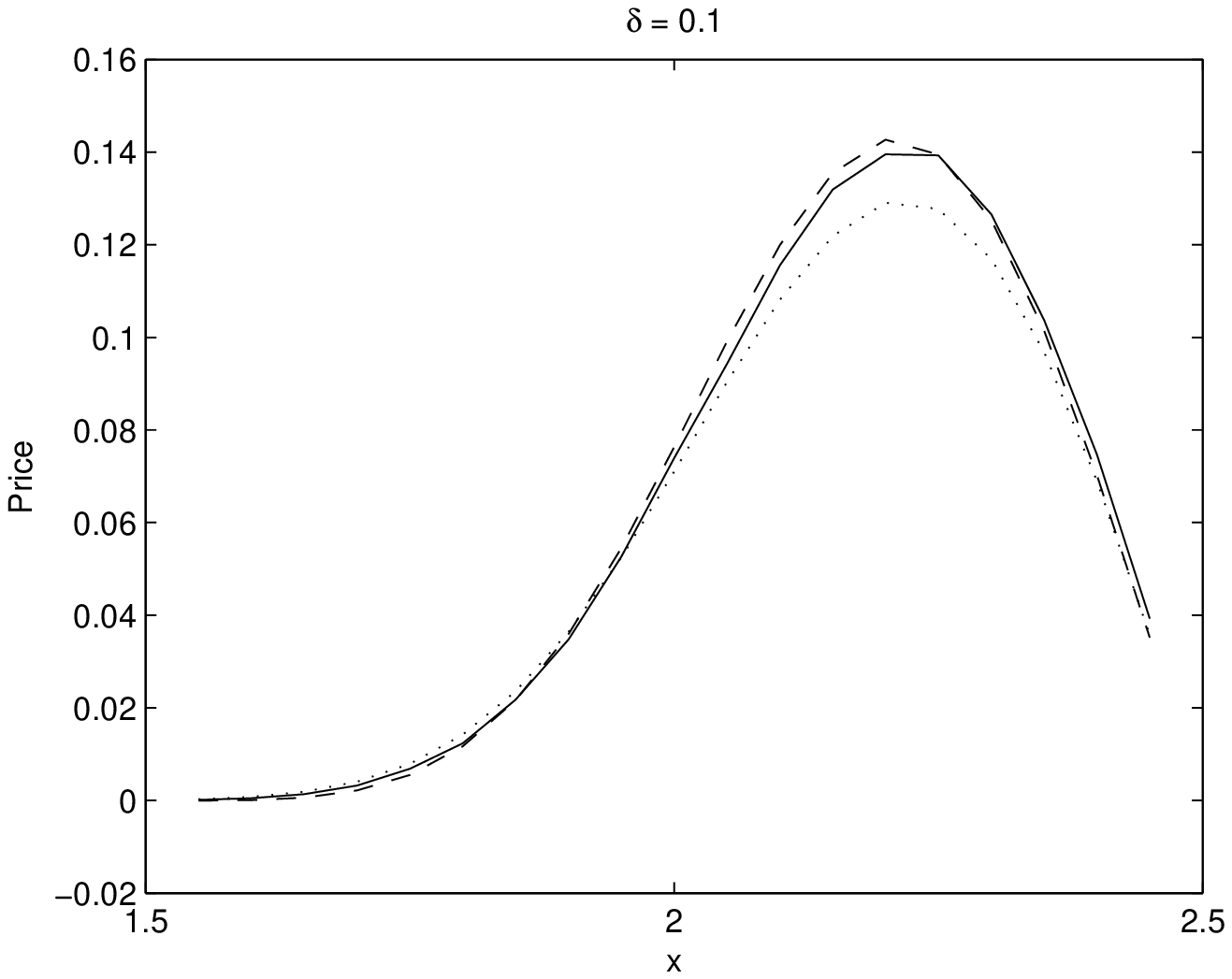} \\ \hline
\includegraphics[width=.5\textwidth,height=.25\textheight]{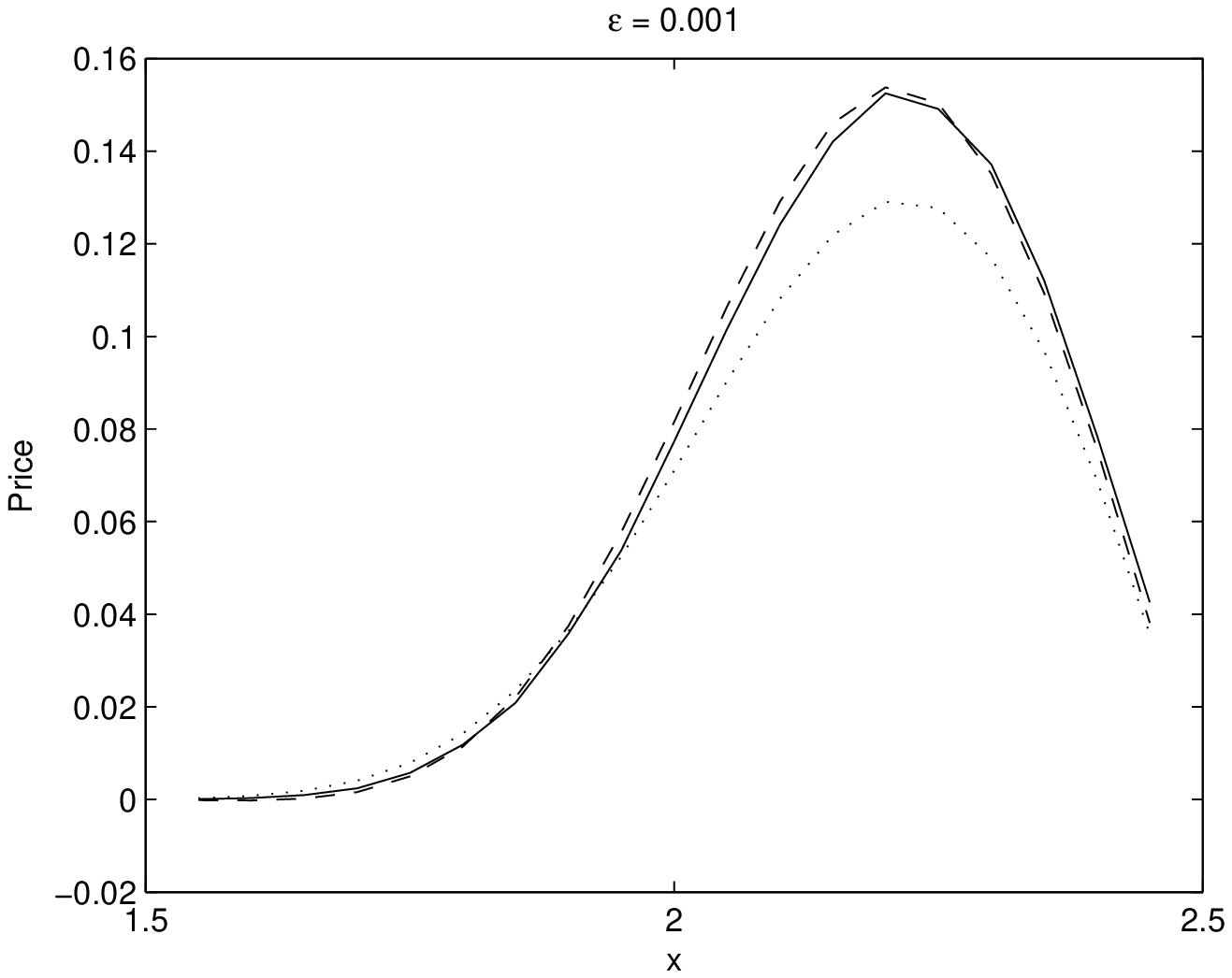} &
\includegraphics[width=.5\textwidth,height=.25\textheight]{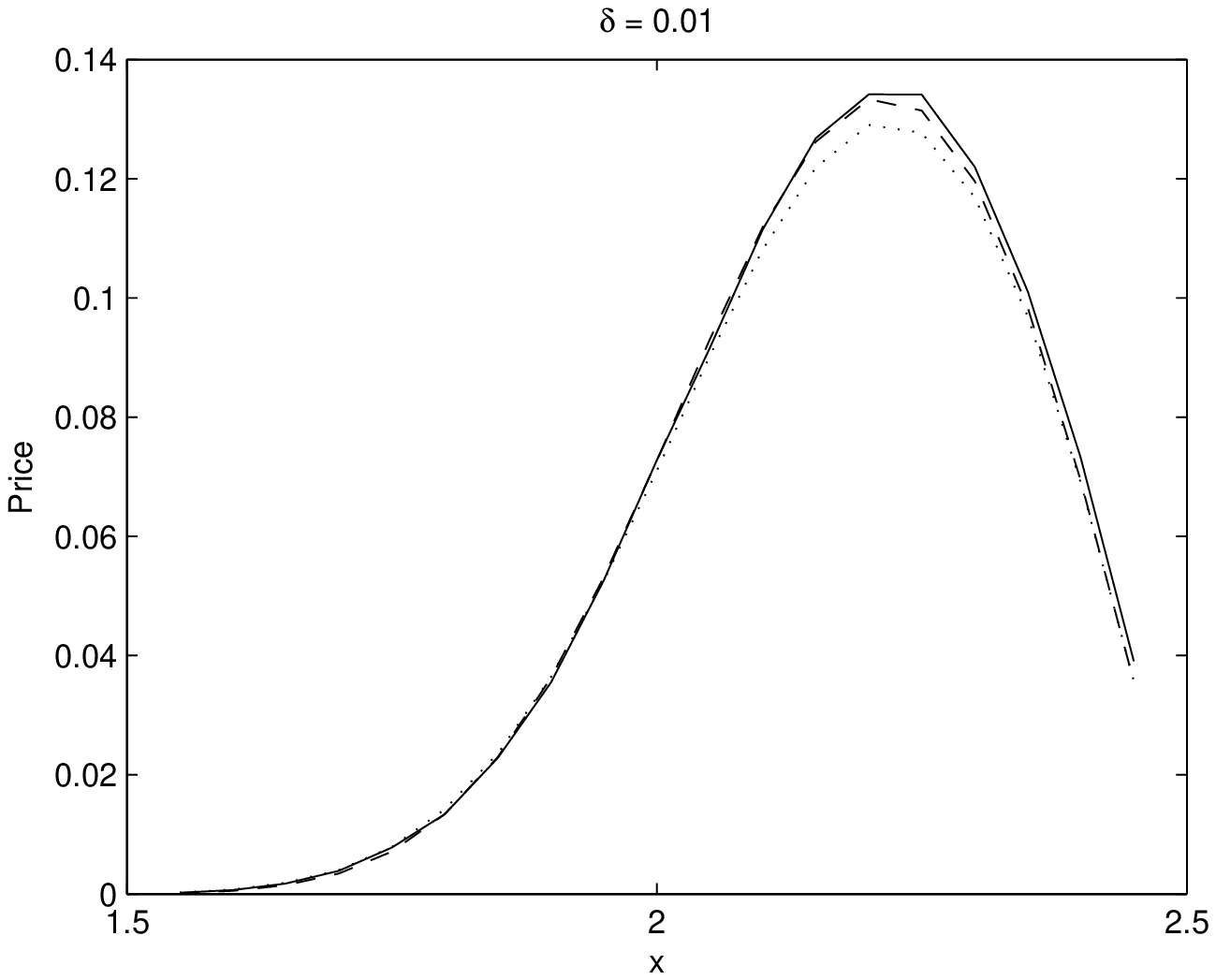} \\ \hline
\end{tabular}
\caption{The price of a double-barrier call option is plotted as a function of the underlying $x$.  On the left we consider the Black-Scholes model with only a fast-varying factor of volatility $Y$ whose dynamics are given by \ref{eq:Yexample}.  On the right, we consider the Black-Scholes model with only a slow-varying factor of volatility $Z$ whose dynamics a given by \ref{eq:Zexample}.  In each plot, the solid black line corresponds to the full price of the option, the dashed line corresponds to our approximation, and the dotted line corresponds to the Black-Scholes price.  For the plots on the left we use parameters $t=1/12$, $y=0$, $r=0.05$, $\sigma=0.34$, $\rho_{xy}=-0.5$, $\beta=1$, $L=1.5$, $K=2.0$, $R=2.5$.  For the plots on the right we use parameters $t=1/12$, $z=2$, $r=0.05$, $\sigma=0.34$, $\rho_{xz}=-0.5$, $g=2$, $L=1.5$, $K=2.0$, $R=2.5$.}
\label{fig:DBslow}
\end{figure}


\begin{figure}
\centering
\begin{tabular}{ | c | c | }
\hline
\includegraphics[width=.5\textwidth,height=.25\textheight]{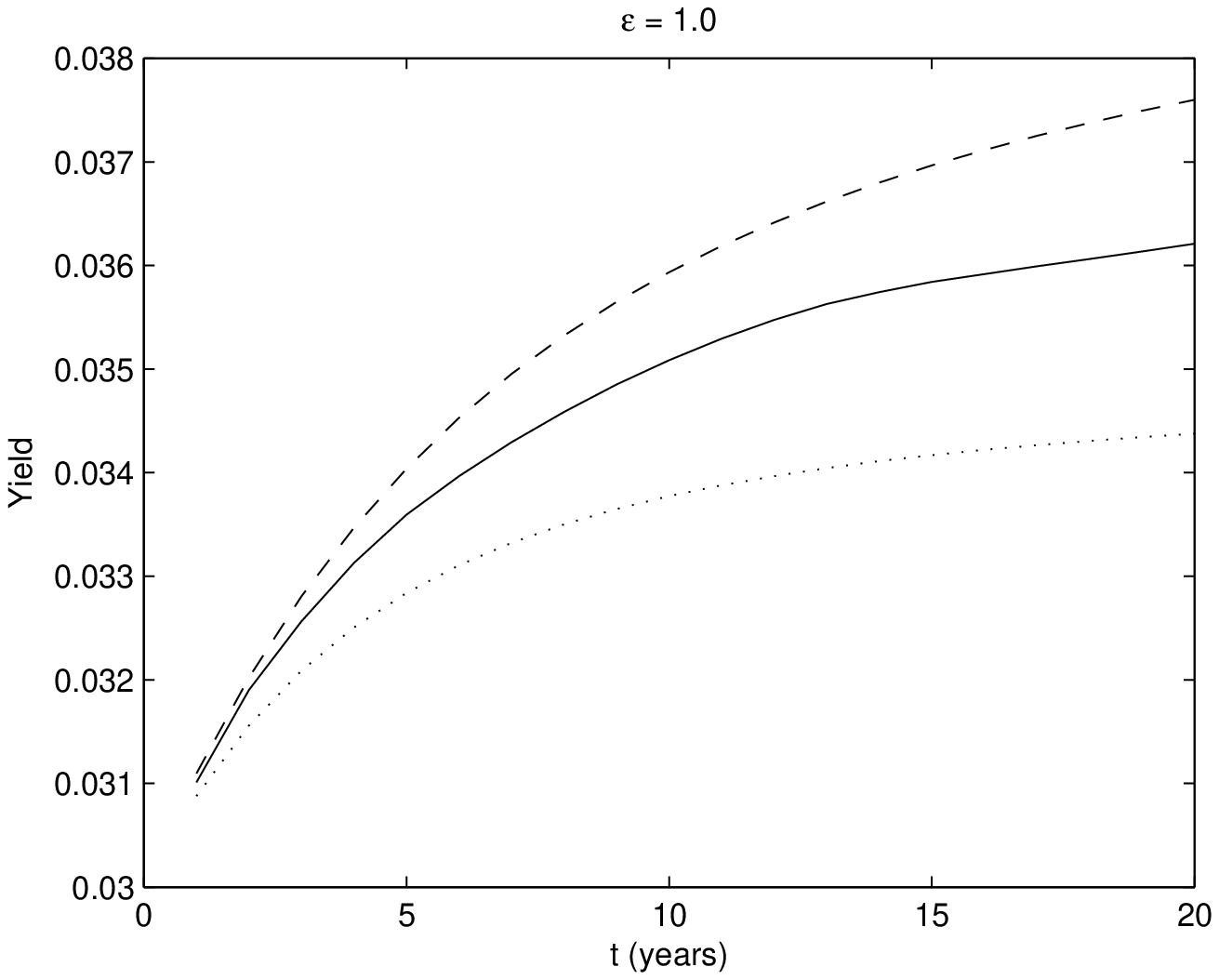} &
\includegraphics[width=.5\textwidth,height=.25\textheight]{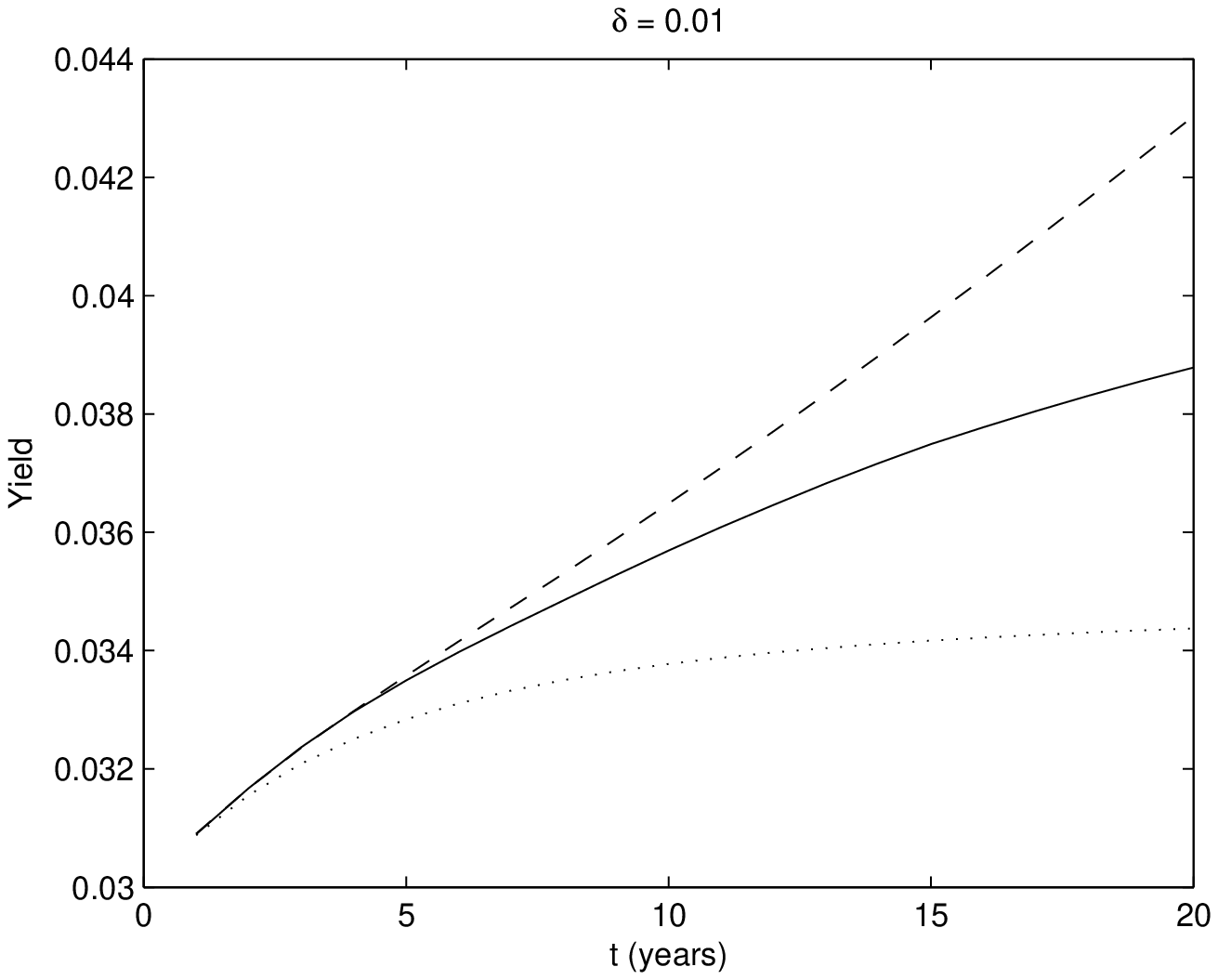} \\ \hline
\includegraphics[width=.5\textwidth,height=.25\textheight]{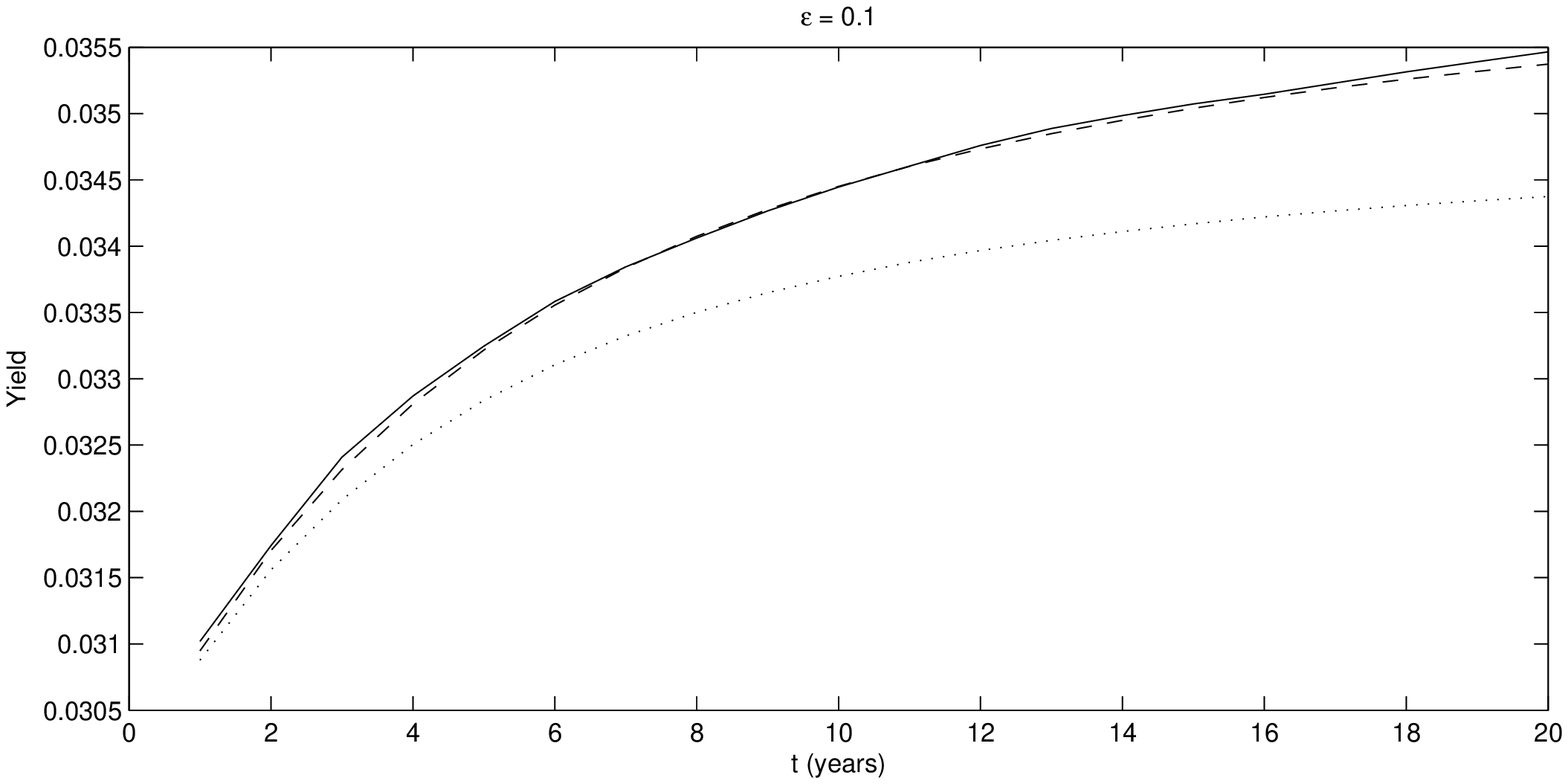} &
\includegraphics[width=.5\textwidth,height=.25\textheight]{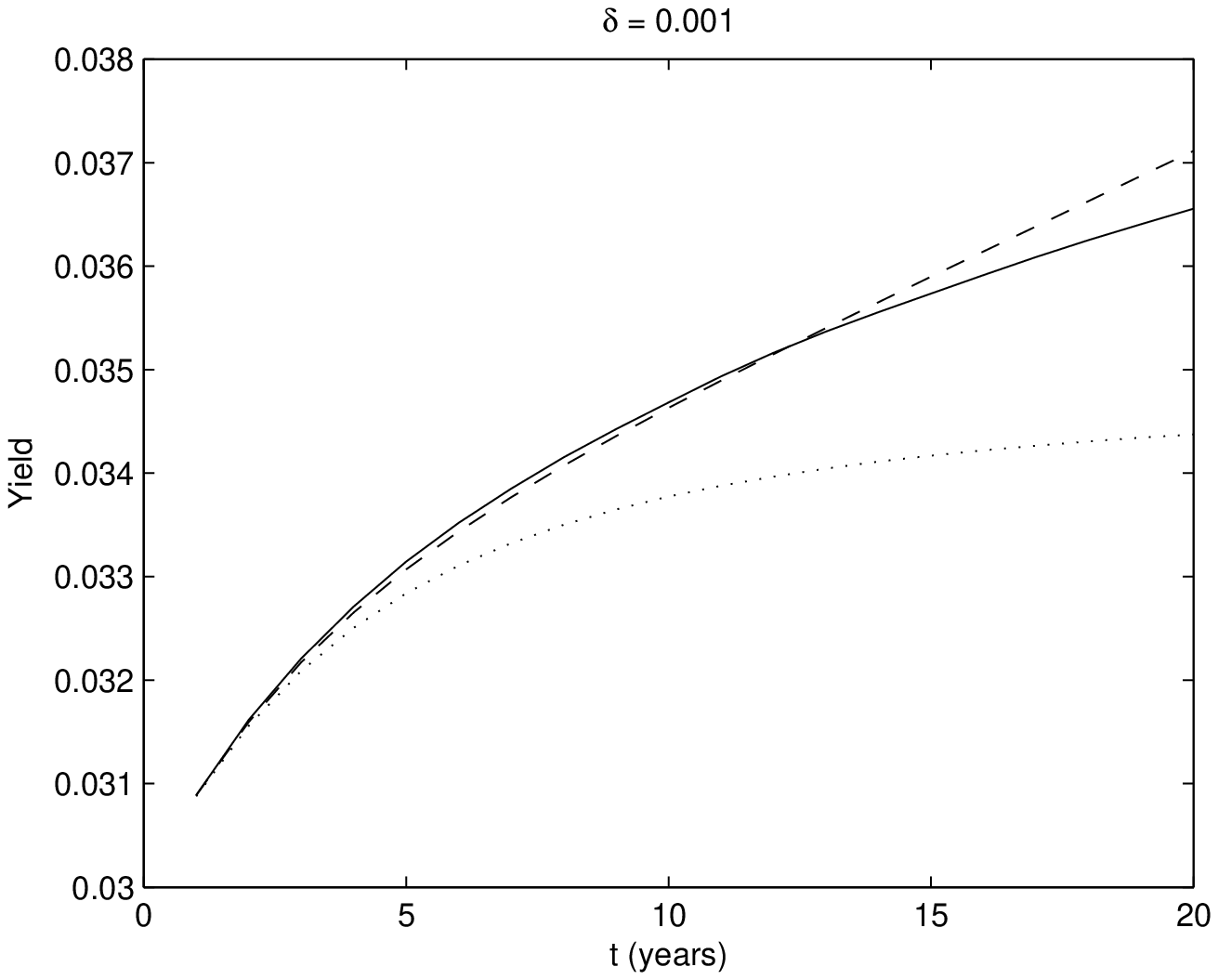} \\ \hline
\includegraphics[width=.5\textwidth,height=.25\textheight]{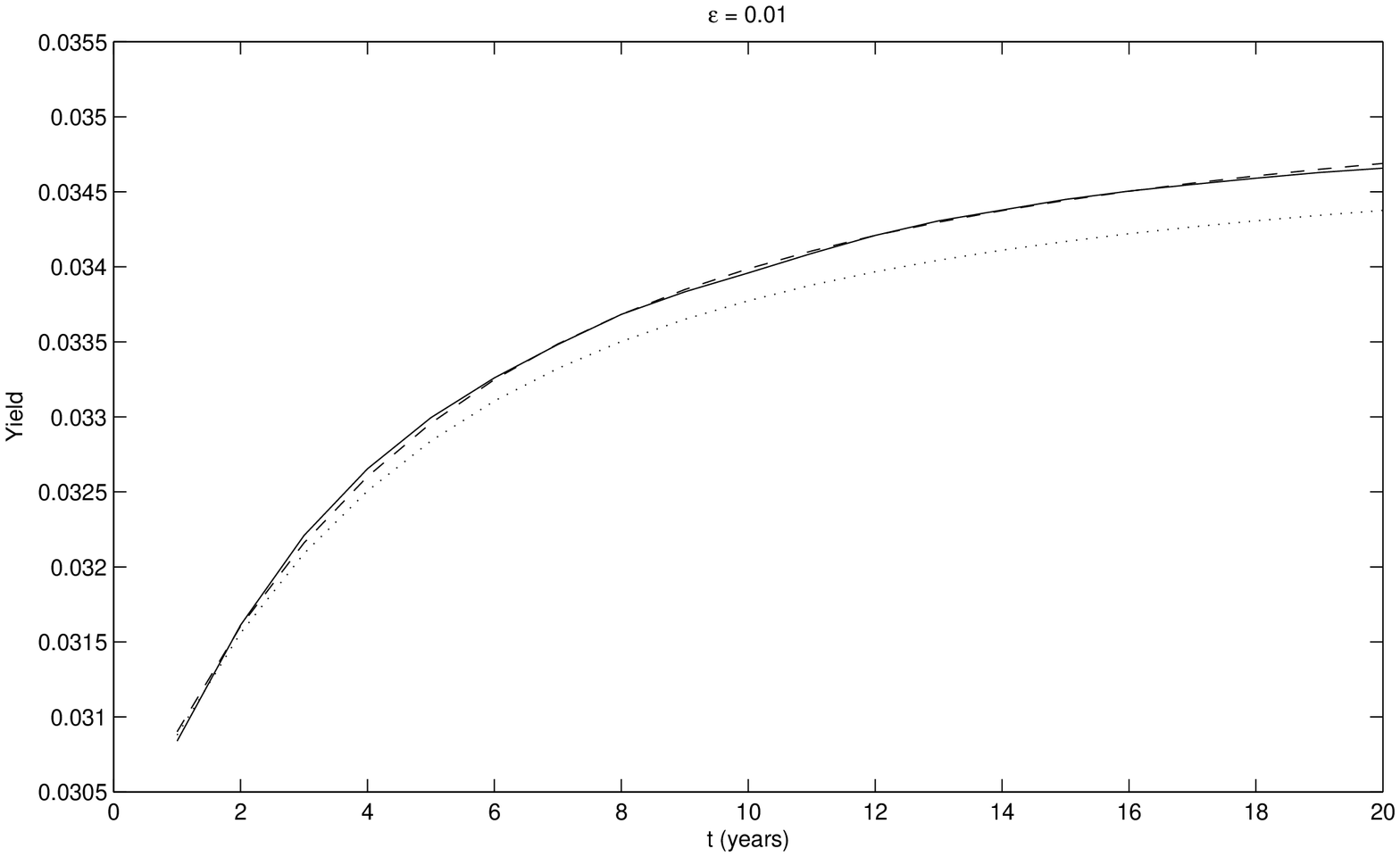} &
\includegraphics[width=.5\textwidth,height=.25\textheight]{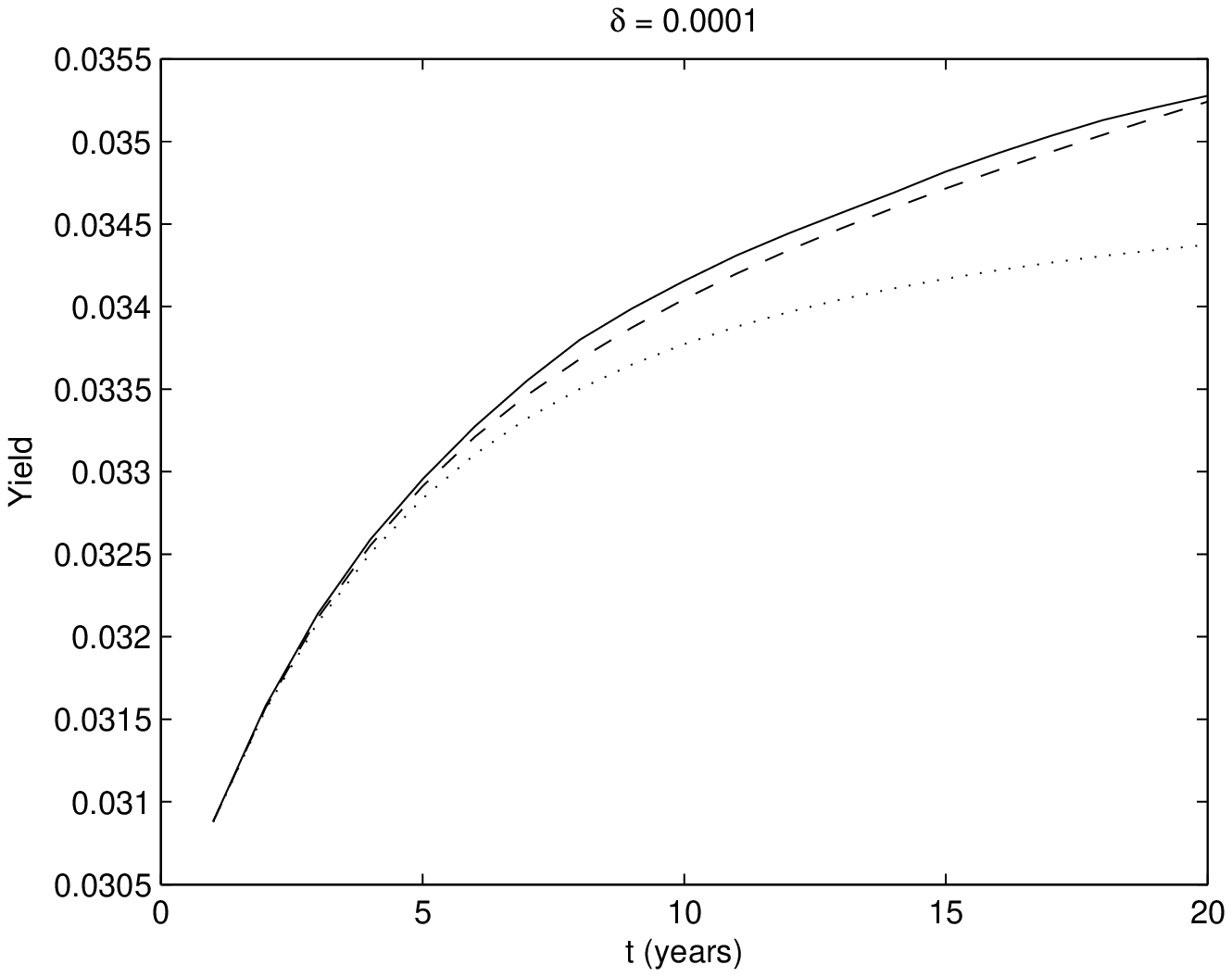} \\ \hline
\end{tabular}
\caption{The yield of a zero coupon bond is plotted as a function of the time to maturity $t$.  On the left we consider the Vasicek model with only a fast-varying factor of volatility $Y$ whose dynamics are given by \ref{eq:Yexample}.  On the right, we consider the Vasicek model with only a slow-varying factor of volatility $Z$ whose dynamics a given by \ref{eq:Zexample}.  In each plot, the solid black line corresponds to the full yield of the bond, the dashed line corresponds to our approximation, and the dotted line corresponds to the Vasicek yield.  For the plots on the left we use parameters $x=0.03$, $y=0$, $\theta=0.05$, $\sigma=0.02$, $\rho_{xy}=-0.5$, $\beta=1$, $\Om = 0.1\,e^{\beta^2/4}$.  For the plots on the right we use parameters $x=0.03$, $z=1.0$, $\theta=0.05$, $\sigma=0.02$, $\rho_{xz}=-0.5$, $g=1$, $\Om = 0.1$.}
\label{fig:Vasicekslow}
\end{figure}


\begin{figure}
\centering
\begin{tabular}{ | c | c | }
\hline
\includegraphics[width=.5\textwidth,height=.25\textheight]{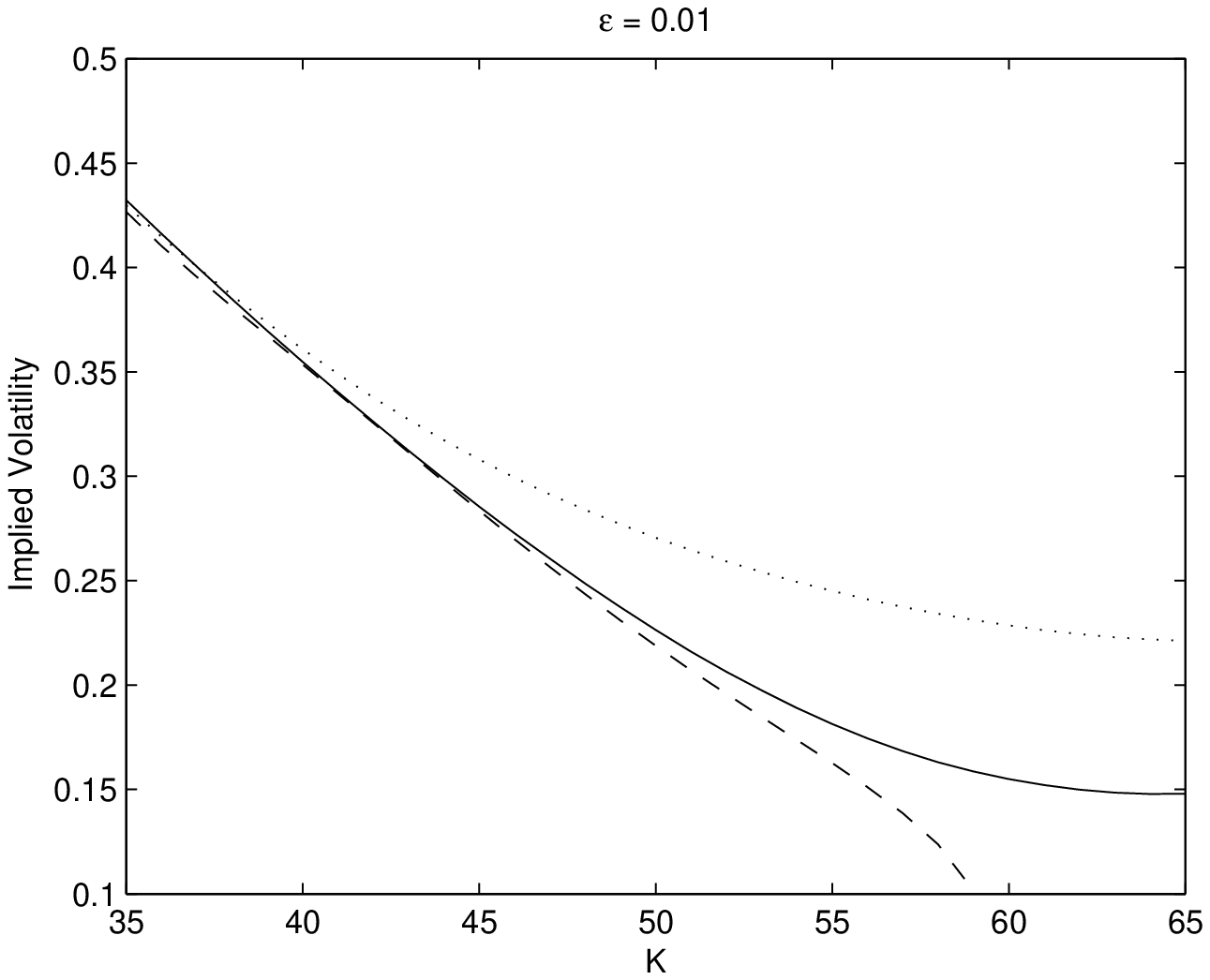} &
\includegraphics[width=.5\textwidth,height=.25\textheight]{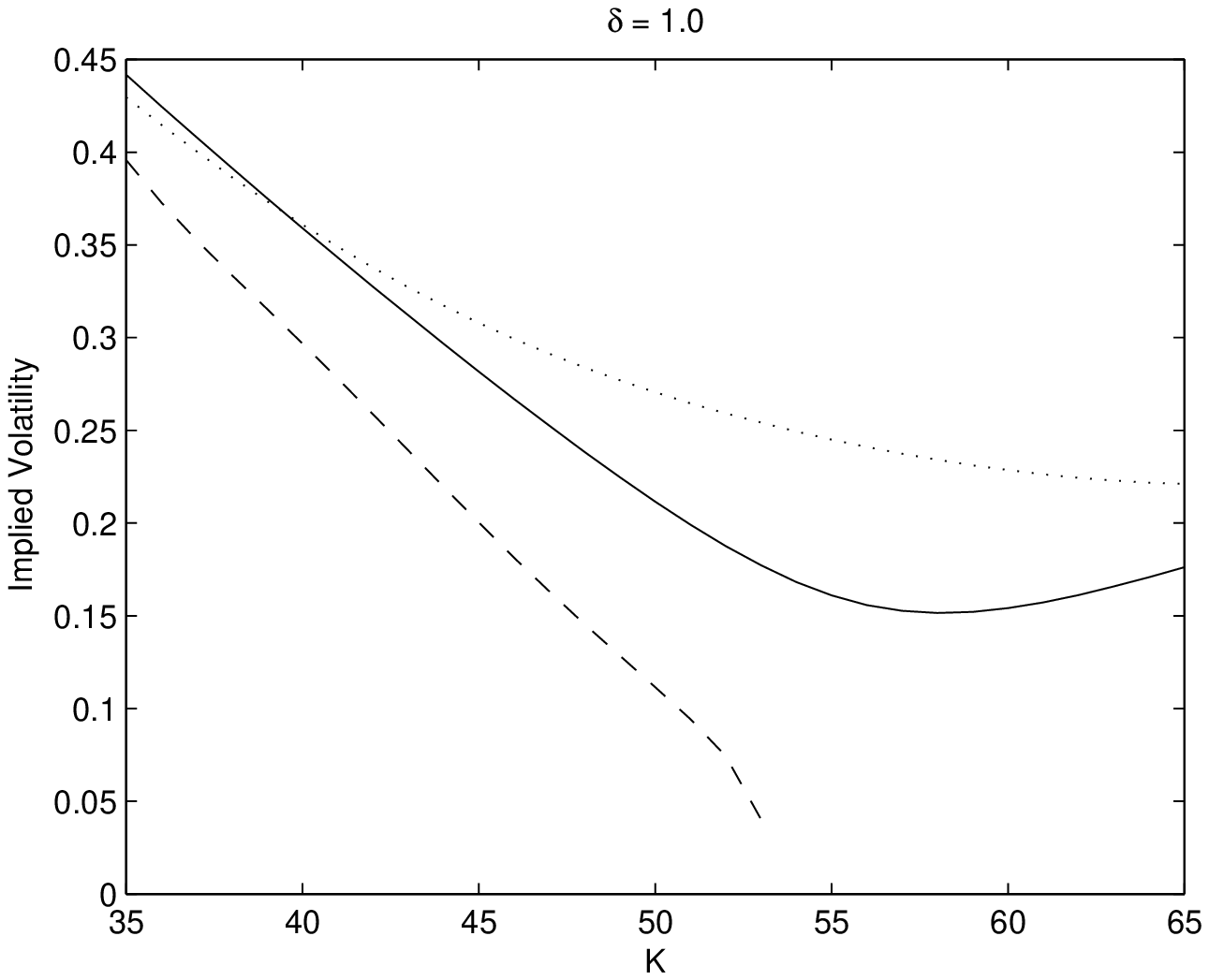} \\ \hline
\includegraphics[width=.5\textwidth,height=.25\textheight]{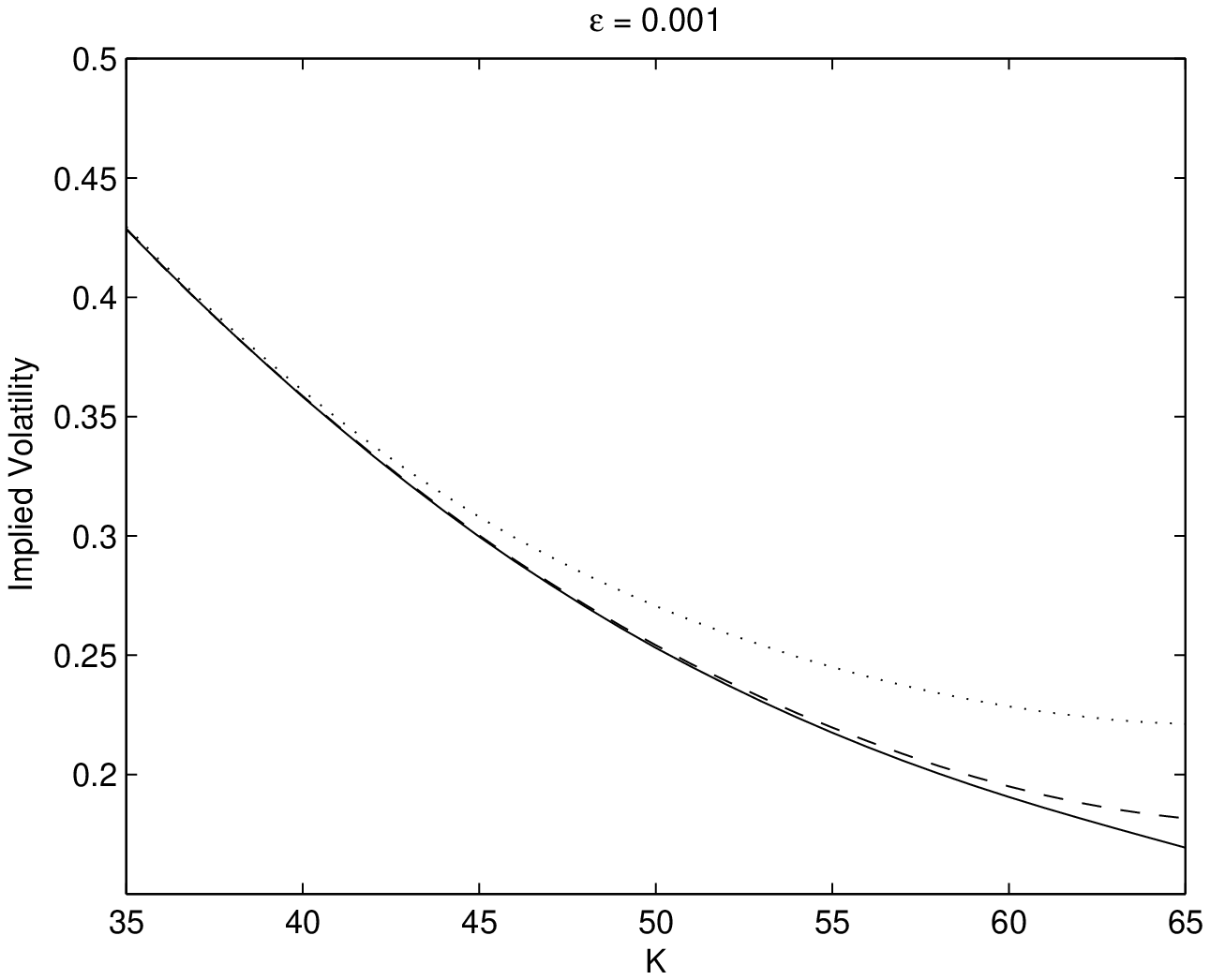} &
\includegraphics[width=.5\textwidth,height=.25\textheight]{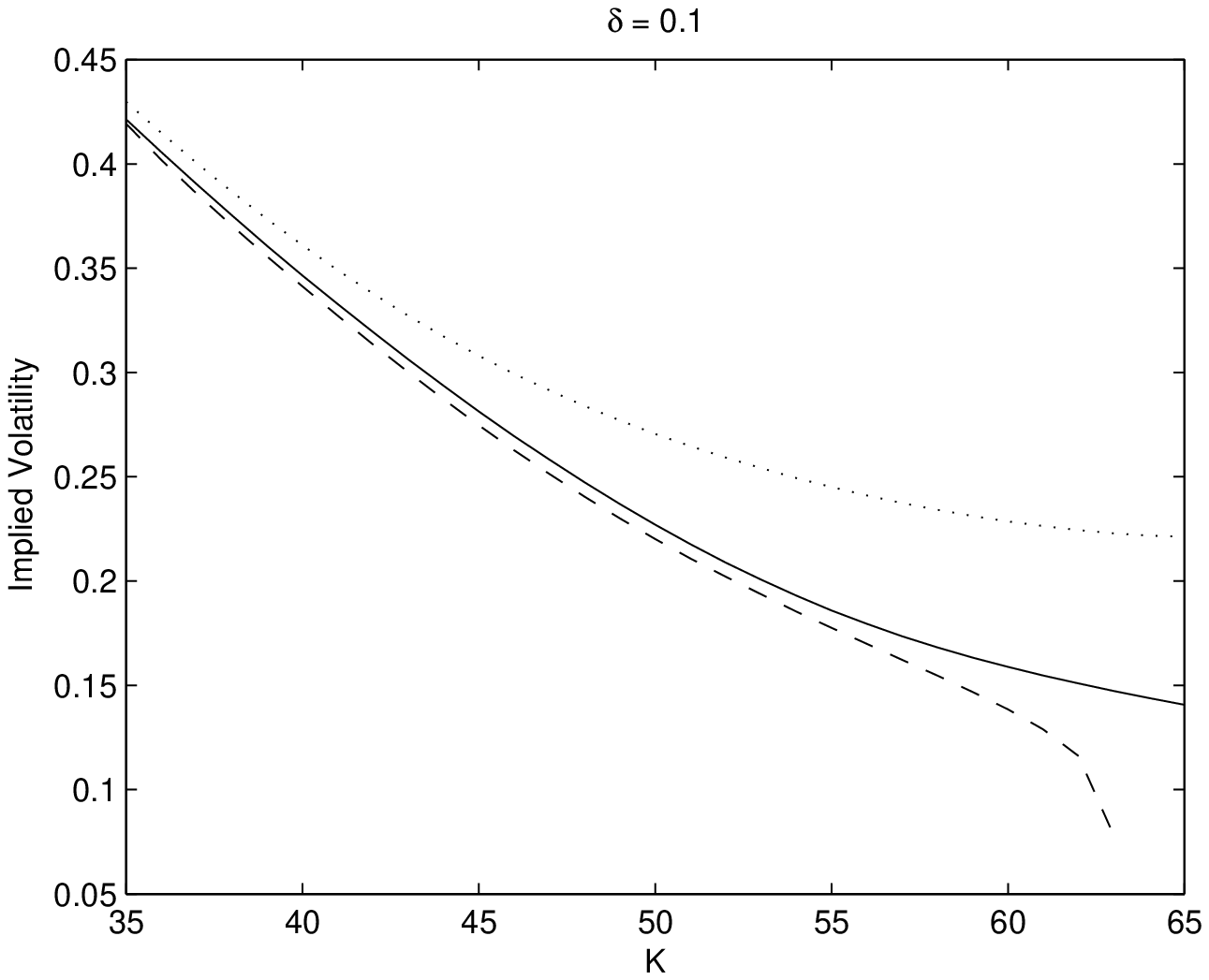} \\ \hline
\includegraphics[width=.5\textwidth,height=.25\textheight]{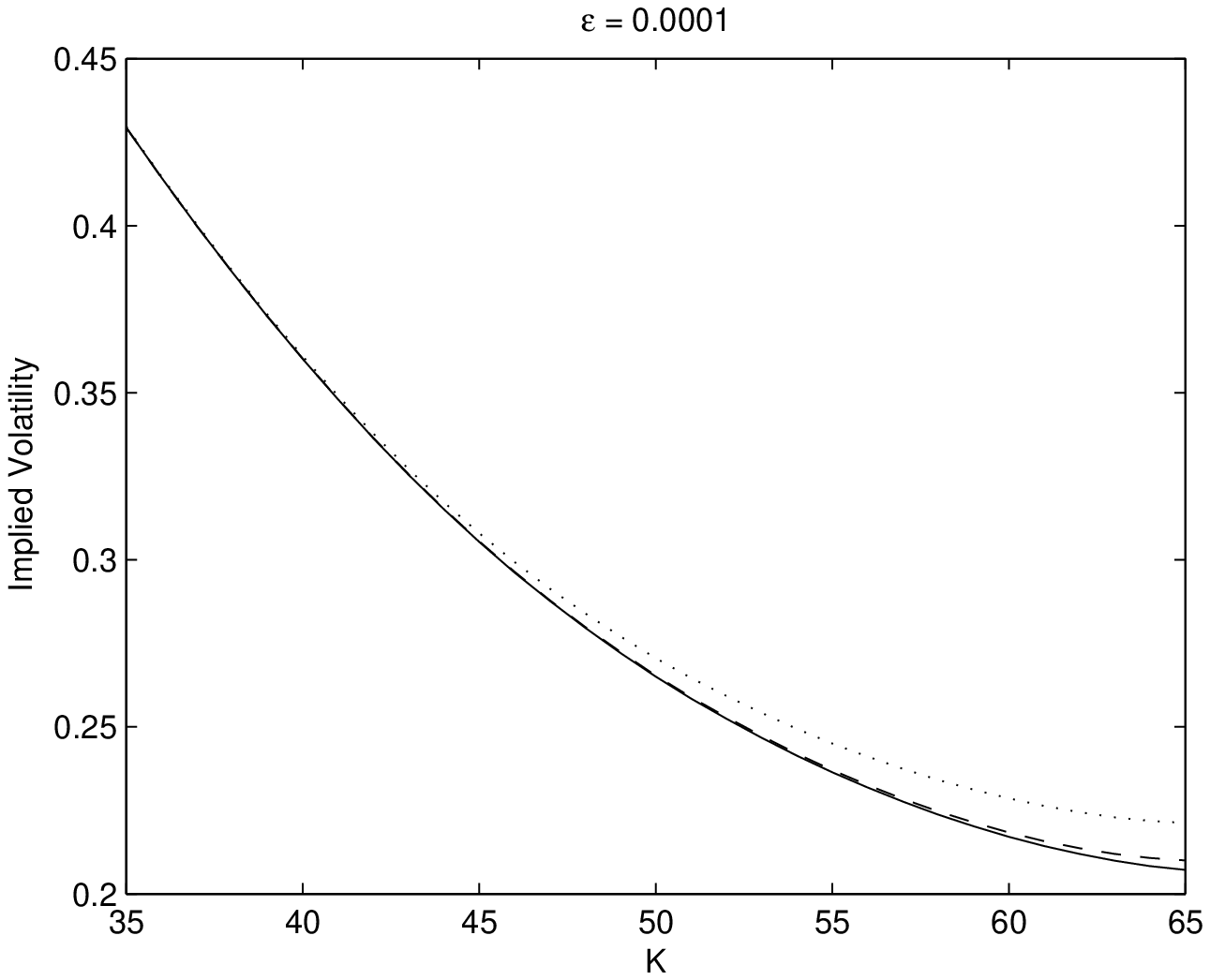} &
\includegraphics[width=.5\textwidth,height=.25\textheight]{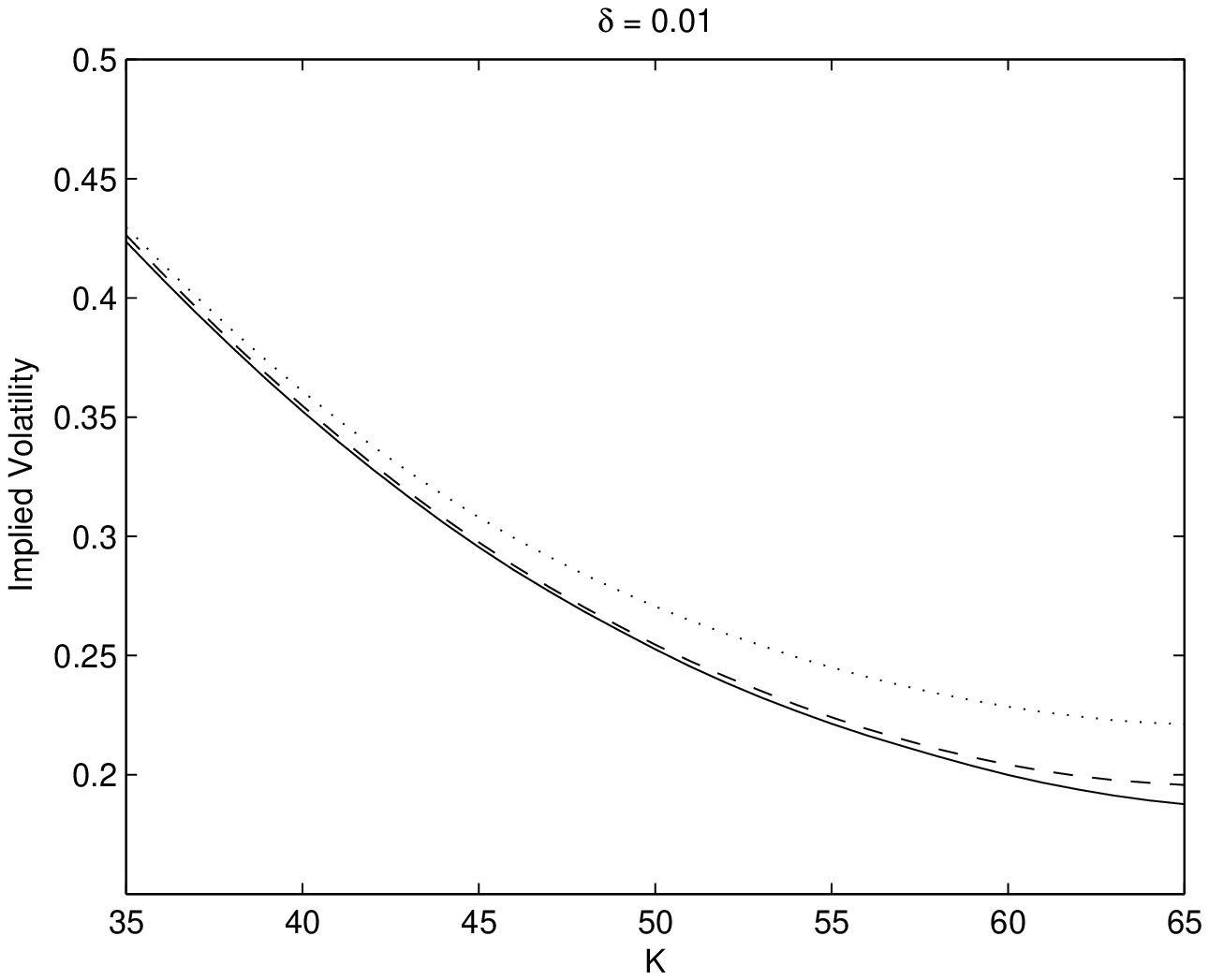} \\ \hline
\end{tabular}
\caption{The implied volatility of a European put option is plotted as a function of the strike price $K$.  On the left we consider a JDCEV model with only a fast-varying factor of volatility $Y$ whose dynamics are given by \ref{eq:Yexample}.  On the right, we consider the JDCEV model with only a slow-varying factor of volatility $Z$ whose dynamics a given by \ref{eq:Zexample}.  In each plot, the solid black line corresponds to the full implied volatility, the dashed line corresponds to our approximation, and the dotted line corresponds to the JDCEV implied volatility.  For the plots on the left we use parameters $t=1$, $x=50$, $\mu=0.05$, $\sig=10$, $\eta=-1$, $c=0.5$, $\rho_{xy}=-0.5$, $y=0$ and $\beta=2$.  For the plots on the right we use parameters $t=1$, $x=50$, $\mu=0.05$, $\sig=10$, $\eta=-1$, $c=0.5$, $\rho_{xz}=-0.5$, $z=2$ and $g=2$.}
\label{fig:JDCEVslow}
\end{figure}

\end{document}